\documentclass[aps,pra,reprint,superscriptaddress,nofootinbib]{revtex4-2}
\usepackage{graphicx}
\usepackage{mathpazo}
\usepackage{dcolumn}
\usepackage{xcolor}
\usepackage{bm}
\usepackage{braket}
\usepackage{lipsum}
\usepackage{amsmath, amsthm, amsfonts, amssymb, bbm, dsfont, quantikz, ragged2e} 
\usepackage{enumitem}
\usepackage{mathtools}
\usepackage{nicematrix}
\usepackage{pifont}
\usepackage[most]{tcolorbox}
\usepackage[dvipsnames]{xcolor}
\usepackage{nicematrix} 
\usepackage[colorlinks,linkcolor=blue,urlcolor=blue, citecolor=blue]{hyperref}

\usepackage{tikz}
\usetikzlibrary{quantikz2}

\usetikzlibrary{decorations}

\usepackage{algorithm}
\usepackage{algpseudocode}

\usepackage{centernot}

\usepackage{tikz}
\usetikzlibrary{quantikz2}

\newtheorem{theorem}{Result}
\newtheorem{lemma}{Lemma}

\newcommand{\CY}[1]{{\color{red}#1}}

\newcommand{\ketbra}[1]{\vert #1 \rangle \langle #1 \vert}
\newcommand{\ketbratwo}[2]{\vert #1 \rangle \langle #2 \vert}
\newcommand{\onenorm}[1]{ \| #1 \|_1}

\definecolor{teal}{RGB}{42, 157, 143}
\definecolor{yellow}{RGB}{233, 196, 106}
\definecolor{red}{RGB}{210, 66, 66}
\definecolor{lred}{RGB}{222,94,100}
\definecolor{lblue}{RGB}{179, 235, 242 }

\newcommand{\fref}[1]{\textcolor{blue}{\hyperref[#1]{Fig.$\,$\bfseries\ref{#1}}}}
\newcommand{\lemref}[1]{\textcolor{blue}{\hyperref[#1]{Lemma$\,$\bfseries\ref{#1}}}}
\newcommand{\thmref}[1]{\textcolor{blue}{\hyperref[#1]{Thm.$\,$\bfseries\ref{#1}}}}

\newcommand{\modesuper}[1]{\scriptscriptstyle(#1)}
\newcommand{\recsuper}[1]{\scriptscriptstyle #1}

\usepackage{silence}
\newcommand{\drawCustomCircle}[8]{%
    \begin{tikzpicture}
        \draw[draw=#4, fill=#5, fill opacity=#6, line width=#7]
            (#1,#2) circle (#3); 
            \node at (#1,#2) {#8};
    \end{tikzpicture}%
}

\newcommand{\drawCustomCircletwo}[8]{%
    \tikz[baseline=(circle_node.base)] {
        \draw[draw=#4, fill=#5, fill opacity=#6, line width=#7]
            (0,0) circle (#3);
        \node[circle, minimum size=2*#3, inner sep=0pt, outer sep=0pt] (circle_node) at (0,0) {#8};
    }%
}

\tikzset{
    every node/.style={font=\small},
    arrow/.style={-{Stealth}, thick},
    implies/.style={->, double equal sign distance, thick}
}

\newtcolorbox[auto counter]{mybox}[2][]{
    breakable = false,
    enhanced,
    sharp corners,
    colback=violet!3!white,
    colframe=violet!40!white,
    fonttitle=\bfseries,
    title={\centering \strut #2}, 
    enlarge bottom at break by=5mm,
    enlarge top at break by=5mm,
    overlay first={%
        \draw[black, line width=0.5mm](frame.south west)--(frame.south east);
        \node[anchor=north east] at (frame.south east) {continued on next page};
    },
    overlay middle={%
        \draw[black, line width=0.5mm](frame.south west)--(frame.south east);
        \draw[black, line width=0.5mm](frame.north west)--(frame.north east);
        \node[anchor=north east] at (frame.south east) {continued on next page};
        \node[anchor=south west] at (frame.north west) {continued from next page};
    },
    overlay last={%
        \draw[black, line width=0.5mm](frame.north west)--(frame.north east);
        \node[anchor=south west] at (frame.north west) {continued from last page};},
    #1
}

\let\oldaddcontentsline\addcontentsline
\renewcommand{\addcontentsline}[3]{}

\begin{document}

\preprint{APS/123-QED}

\title{Unspeakable Coherence Concentration}

\author{Benjamin Stratton}
\email[]{ben.stratton@bristol.ac.uk}
\affiliation{Quantum Engineering Centre for Doctoral Training, H. H. Wills Physics Laboratory and Department of Electrical \& Electronic Engineering, University of Bristol, BS8 1FD, UK}
\affiliation{H.H. Wills Physics Laboratory, University of Bristol, Tyndall Avenue, Bristol, BS8 1TL, UK}

\author{Chung-Yun Hsieh}
\affiliation{H.H. Wills Physics Laboratory, University of Bristol,
Tyndall Avenue, Bristol, BS8 1TL, UK}

\author{Paul Skrzypczyk}
\affiliation{H.H. Wills Physics Laboratory, University of Bristol,
Tyndall Avenue, Bristol, BS8 1TL, UK}

\date{\today}

\begin{abstract}
Unspeakable coherence is a key feature separating quantum and classical physics. Modelled as asymmetry with respect to a continuous transformation generated by a physically relevant observable, such as the Hamiltonian or angular moment, unspeakable coherence has been shown to be the relevant notion of coherence for achieving quantum advantage in the tasks of metrology, reference frame alignment and work extraction, among others. A question of both practical and foundational value is: Given some copies of a state with low coherence, can we prepare a more coherent state via coherence non-increasing operations? Here, we study this question in the minimal limiting case: Given two uncorrelated copies of a coherent state, can one, via globally coherence non-increasing unitaries, increase the coherence in a subsystem? We fully solve this problem for qubits, identifying the optimal unitaries and revealing the existence of bound coherence. This is then used to create a completely constructive multi-qubit coherence enhancement protocol, where only effective-qubit unitaries are used. Unexpectedly, in this protocol, we show that there exists states for which the ratio of the input-output coherence can be amplified unboundedly. Extending beyond qubits, we derive two fundamental upper bounds on the amount of local coherence that can be increased and prove a no-go theorem showing that certain global correlations cannot be converted to local coherence.

\end{abstract}

\maketitle
\section{Introduction}
Coherence is a key quantum feature.  
Whilst it is necessary for many quantum-enhanced tasks --- such as those in metrology~\cite{PhysRevLett.96.010401}, communication~\cite{Chen_2021}, computation~\cite{nielsen_chuang_2010}, and more --- it is highly non-trivial to generate, maintain, and manipulate. Consequently, as with many other quantum properties, it has proven fruitful to treat it as a resource \cite{Chitambar_2019, PhysRevA.94.052324, PhysRevLett.113.140401}.

Coherence comes in two forms: speakable and unspeakable~\cite{PhysRevA.94.052324, PhysRevLett.113.140401}. In both cases, coherence is only meaningful when considered with respect to a known decomposition of the Hilbert space into subspaces. Colloquially, for unspeakable coherence, this decomposition is performed by considering the eigenspaces of some physically relevant observable --- such as the Hamiltonian, angular momentum, or photon number. Whereas, for speakable coherence, an arbitrary basis is used. 

The notion of unspeakable coherence has attracted significant attention as it provides a natural framework for describing coherence manipulation under symmetry-constrained dynamics, such as those conserving energy or momentum \cite{Lostaglio2015, PhysRevLett.132.200201, PhysRevX.5.021001, PhysRevLett.132.180202, Marvian2020, PhysRevLett.123.020403, PhysRevLett.123.020404, PhysRevLett.113.150402, PhysRevA.103.022403, PhysRevLett.128.240501}. As a result, it has been instrumental in building our understanding of the role of coherence in thermodynamics \cite{Lostaglio2015, PhysRevX.5.021001, PhysRevLett.132.200201}. Moreover, unspeakable coherence has proven to be the relevant notion of coherence in the important tasks of quantum metrology~\cite{PhysRevLett.96.010401}, reference frame alignment~\cite{RevModPhys.79.555}, and work extraction~\cite{Lostaglio2015}, among others~\cite{PhysRevA.93.052331}. Therefore, we here focus on unspeakable coherence, henceforth referring to it as just coherence. 

Given some partially coherent state, the question of whether its coherence can be increased under coherence non-increasing operations is of both foundational and practical interest \cite{PhysRevLett.116.120404, PhysRevLett.132.180202}. Foundationally, it allows us to understand the interplay of coherence and some conserved quantity, such as how coherence can be manipulated by energy-conserving dynamics~\cite{Lostaglio2015, PhysRevX.5.021001, PhysRevLett.132.180202, PhysRevLett.132.200201, Gour_2008}. Practically, these insights can guide the development of methods for generating highly coherent states, which are an essential requirement for quantum technologies.

In this work, we look to answer the resource concentration problem~\cite{hsieh2024informationalnonequilibriumconcentration} for coherence: given two uncorrelated copies of a coherent state, can one, via globally coherence non-increasing unitaries i.e., closed system dynamics, increase the coherence in a subsystem? This is the minimally feasible, and therefore the most practically implementable, scenario where coherence enhancement can be considered. Whilst the manipulation and enhancement of unspeakable coherence has been well studied in the asymptotic and catalytic conversion scenarios~\cite{Marvian2020, PhysRevLett.132.180202, PhysRevLett.123.020403, PhysRevLett.123.020404, PhysRevLett.113.150402, PhysRevA.103.022403, PhysRevLett.128.240501, PhysRevLett.132.200201}, an understanding of this minimalist case --- where access to almost all available freedoms is limited --- is missing. Understanding this limiting case sets a lower bound on what coherence enhancement is possible, complementing the upper bounds set by considering asymptotic distillation protocols \cite{PhysRevLett.132.180202}. Moreover, it allows a bottom-up understanding of what is necessary for coherence enhancement. 

Here, we fully solve the problem in the qubit case and use this complete characterisation to build a fully constructive protocol for multi-qubit coherence amplification, where only effective qubit unitaries are used. Unexpectedly, in this protocol, {we find that the input-output ratio of the unspeakable coherence can be unboundedly increased for some input states}. We then go beyond qubits to provide multiple, incomparable, fundamental limitations as well as no-go results on coherence enhancement.

\section{Unspeakable Coherence}
Consider a finite-dimensional system and let $L$ be a Hermitian operator acting on it, such that $L$ describes a physical observable. A state $\rho$ is said to be coherent with respect to $L$ if and only if $[\rho, L] \neq 0$~\cite{PhysRevA.94.052324}, as this condition is equivalent to $\rho$ containing superpositions between different eigenspaces of $L$. If this is not the case, we say $\rho$ is incoherent. This is equivalent to \hbox{$e^{-iLx}\rho e^{iLx}=\rho$ $\forall\,x\in\mathbb{R}$}, meaning that $\rho$ is `translation-invariant' under the continuous translation generated by $L$~\cite{PhysRevA.94.052324}.

As we aim to study coherence enhancement from one system to another, it is essential to characterise operations that cannot generate coherence. This ensures that no additional coherence can be added into the system. Formally, these are called {\em allowed operations} which are completely-positive trace-preserving linear maps, $\mathcal{E}$, satisfying
\begin{equation}
\mathcal{E}\left[e^{-iLx}(\cdot)e^{iLx}\right]=e^{-iLx}\mathcal{E}(\cdot)e^{iLx}\quad
    \forall ~ x \in \mathbb{R};
    \label{allowedOperations}
\end{equation}
namely, they are translationally covariant with respect to $L$. If one considers only unitary dynamics $\mathcal{E}(\cdot)=U(\cdot)U^{\dagger}$, Eq.~\eqref{allowedOperations} simplifies to \hbox{$[U, e^{-iLx}] = 0$ $\forall~x \in \mathbb{R}$}, which holds if and only if $[U, L]=0$. Hence, a unitary $U$ gives an allowed operation, termed an {\em allowed unitary}, if and only if it is block diagonal with respect to the eigenspaces of $L$, i.e., \hbox{$V = \sum_{k} V_k \Pi_k$}, where each $V_k$ is an independent unitary and $\{\Pi_k\}_k$ are projectors onto the individual eigenspaces of $L$ such that $V_kV_l=\delta_{k,l} \Pi_k$. From here, it is easy to see that these allowed unitaries cannot generate coherence.

We can apply this approach to a bipartite system $AB$ by considering $L_{AB} = L_A \otimes \mathbb{I}_B + \mathbb{I}_A \otimes L_B$ (capital letters in subscripts denote separate spaces; $\mathbb{I}$ is the identity operator). If $[\rho_{AB}, L_{AB}]=0$, then the bipartite state $\rho_{AB}$ is incoherent with respect to the eigenspaces of $L_{AB}$; also, a bipartite allowed unitary $V_{AB}$ satisfies \hbox{$[V_{AB}, L_{AB}]=0$}
(see Supplementary Material (SM)~\ref{appendix:globalVsLocal} for further discussions). If $L_{AB}$ is non-degenerate, then all allowed unitaries are diagonal in the joint eigenbasis of $L_{AB}$. Such unitaries do not facilitate any interaction between the two subsystems, and no coherence concentration is possible. If $L_{AB}$ contains degeneracies, then the allowed unitaries are arbitrary unitaries within the degenerate subspaces of $L_{AB}$. Physically, this means that no matter what property $L_{AB}$ represents, it can be freely exchangeable between $A$ and $B$, as long as its total value is conserved. E.g., if $L_{AB}$ is the total Hamiltonian, then $[U_{AB}, L_{AB}]=0$ ensures the total energy is conserved under the dynamics $U_{AB}$. If one can move energy between $A$ and $B$ whilst conserving total energy, then coherence can be locally increased; otherwise, it cannot.  

From now, we consider our {\em local} physical observables to be the non-degenerate, truncated number operators $L_X = \sum_{n=0}^{d-1} n \ketbra{n}_X$ for $X=A,B$. This ensures that coherence cannot be locally increased, but that $L_{AB}$ contains degeneracies so that coherence concentration is possible. 

\section{Modes of Coherence}
Under translationally covariant operations, quantum states can be decomposed into a set of modes \cite{PhysRevA.90.062110}. Each mode is a set of components of the state that transform independently due to the symmetry of the covariant dynamics. Formally, the {\em $j$th mode} of a $d$-dimensional state $\rho$ is given by 
\begin{equation}
\rho^{\scriptscriptstyle(j)} \coloneqq \sum_{n=0}^{d-1}\proj{n+j} \rho \proj{n},
\end{equation}
where we take $\ket{n+j}=0$ if $n+j\ge d$. This is an operator supported in the {\em $j$th mode subspace} \hbox{$\mathcal{B}^{(j)} \coloneqq {\rm span}\{\ketbratwo{n+j}{n}\}_{n=0}^{d-1}$}. By definition, we have $-d+1\le j\le d-1$ and $\rho = \sum_{j=-d+1}^{d-1} \rho^{\scriptscriptstyle(j)}$. We note that since $\rho$ is a state, $[\rho^{(j)}]^{\dagger} = \rho^{(-j)}$ $\forall\,j$, meaning it suffices to focus on $\rho^{(j)}$'s with $j\ge0$. In addition, it can be seen that $\rho=\rho^{(0)}$ if and only if $\rho$ is incoherent, i.e., $[\rho, L]=0$. Finally, given two states $\rho$ and $\sigma$, if \hbox{$\sigma = \mathcal{E}(\rho)$} for some allowed operation $\mathcal{E}$, then $\sigma^{\scriptscriptstyle(j)} = \mathcal{E}(\rho^{\scriptscriptstyle(j)})$ $\forall\,j$.
Hence, under symmetry constrained dynamics, the output state's $j$th mode depends only on the input state's $j$th mode. 
Consequently, $\rho^{\scriptscriptstyle(j)}=0$ implies $\sigma^{\scriptscriptstyle(j)}=0$, which can be generalised to \hbox{${\rm modes}(\sigma) \subseteq {\rm modes}(\rho)$}, where ${\rm modes}(\rho)\coloneqq\{j\in\mathbb{N}\,|\,\rho^{(j)}\neq0\}$
is the set of all integers labelling modes for which $\rho$ has a nonzero component. 

In the bipartite setting, the global modes of a bipartite product state are given by a combination of local modes~\cite{PhysRevA.90.062110}:
$
(\rho_A \otimes \rho_B)^{\scriptscriptstyle(j)} = \sum_{k} \rho_A^{\scriptscriptstyle(k)} \otimes \rho_B^{\scriptscriptstyle(j-k)}.
$ It can further be seen that the $j$th mode subspace reads  
$
    \mathcal{B}_{AB}^{\modesuper{j}} \coloneqq {\rm span}\{ \ketbratwo{n+k, m+j-k}{n,m}_{AB}\,|\,n,m,k\in\mathbb{N}\}.
$
Using this, in Appendix A we show the following fact: 
\begin{lemma} \label{global_Modes_to_local_modes}
    Let $\rho_A \coloneqq {\rm tr}_B(\rho_{AB})$ be the reduced state of $\rho_{AB}$ in $A$. Then 
    $
        \rho_{A}^{\scriptscriptstyle(j)} = {\rm tr}_{B}\big( \rho_{AB}^{\scriptscriptstyle(j)} \big),
    $
    meaning the $j$th local mode depends only on the $j$th global mode. 
\end{lemma}

Given the number of available modes is a function of dimension, there are more available modes globally than locally. Hence, Lemma~\ref{global_Modes_to_local_modes} implies that tracing over a global mode with $|j|>d_A-1$ equals zero, with it making no contribution to the local modes in the $A$ subsystem, where $d_{AB}$ and $d_A$ be the dimensions of $AB$ and $A$, respectively. 

\section{Coherence Concentration Problems}
To study coherence concentration, we first quantify the coherence within a specific mode. 
To do this, we use \hbox{$M^{\scriptscriptstyle(j)}(\rho) \coloneqq \onenorm{\rho^{\scriptscriptstyle(j)}}$} \cite{PhysRevA.90.062110},
where $\onenorm{P}\coloneqq{\rm tr}\sqrt{PP^\dagger}$ is the {\em one norm} of the operator $P$. For all $j\neq0$, this family of measures are monotonic under allowed operations, and faithful i.e., \hbox{$ M^{\scriptscriptstyle(j)}(\rho) = 0~\forall j $ iif $[\rho, L]=0$}. It is also faithful for a specific mode: \hbox{$ M^{\scriptscriptstyle(j)}(\rho) = 0$ iif $\rho^{\modesuper{j}}=0$}. For a given integer $j\neq0$ and a state $\rho$, the key figure of merit we will be concerned with is 
\begin{equation}
    \Delta M^{\scriptscriptstyle(j)}_A \coloneqq \max_{V_{AB}}  M^{\scriptscriptstyle(j)}(\sigma_A)- M^{\scriptscriptstyle(j)}(\rho_A),
\end{equation}
which maximises over all allowed unitaries $V_{AB}$ ($[V_{AB}, L_{AB}]=0$) where \hbox{$\sigma_A \coloneqq { \rm tr}_B\big[V_{AB}(\rho_A \otimes \rho_B)V^\dagger_{AB} \big]$}. Our central question then amounts to asking: {\em for a given $j$, is \hbox{$ \Delta M^{\scriptscriptstyle(j)}_A \geq 0$}?}

\section{Optimal Qubit Coherence Concentration}
To better understand the physics, we first focus on performing coherence concentration using two qubits, where we can set the local generator to be \hbox{$L=\ketbra{1}$} without loss of generality. Each qubit has only two modes, with support on the $1$st mode, $\ketbratwo{1}{0}$, indicating the presence of coherence. This $1$st mode then only has a single component---the single off-diagonal element of $\rho$ --- which is the sole parameter of a qubit state that determines its coherence. Hence, any coherence measure on a qubit will depend only on this parameter. Therefore, whilst we consider $\Delta M^{\modesuper{j}}$, our results can be applied to {\em any} qubit coherence measure. 

Any qubit state $\rho$ can be completely characterised by $p_{00}\coloneqq\bra{0} \rho \ket{0}$ and $p_{01}\coloneqq\bra{0} \rho \ket{1}$, such that initially \hbox{$M^{\modesuper{1}}(\rho_A)=\vert p_{01} \vert$}. From here, we can explicitly write down the analytical form for the optimal coherence concentration possible (see Appendix B for sketched proof and SM.~\ref{proof section: result 1} for a full proof). 
\begin{theorem}\label{result_3:maxIncreaseQubits} 
For a qubit state $\rho$, we have
\begin{equation}
    \Delta M_A^{\modesuper{1}} = \vert p_{01} \vert ~ \bigg( \sqrt{1 + (2p_{00} - 1)^2} - 1 \bigg),
\end{equation}
which is achieved by applying a unitary rotation of angle $\cos^{-1}\left(1/\sqrt{1+(2p_{00}-1)^2}\right)$ in the $\{ \ket{01}, \ket{10} \}$ degenerate subspace of $L_{AB}$.  
\end{theorem}
Result~\ref{result_3:maxIncreaseQubits} implies that if $\rho$ has no initial coherence, i.e, $p_{01}=0$, then no coherence can be concentrated, as expected. Interestingly, it can also be seen that if $p_{00}=1/2$ (i.e., the diagonal of $\rho$ is maximally mixed in the eigenbasis of $L$), then no coherence can be concentrated. Put alternately, if $\rho$ sits on the plane through the equator of the Bloch sphere defined by $L$, then no coherence can be concentrated. Hence, qubit states having $p_{00}=1/2$ and $p_{01} > 0$ carry \textit{bound coherence}---they do not have concentratable coherence, despite having non-zero and non-maximal initial coherence~\cite{hsieh2024informationalnonequilibriumconcentration}.  

\section{Qubit Concatenation} 
Result~\ref{result_3:maxIncreaseQubits} fully solves the coherence concentration problem for two qubits. We now assume access to many copies of a qubit, but still aim to maximally increase the coherence of single qubit. Whilst this puts us outside of the resource concentration framework \cite{hsieh2024informationalnonequilibriumconcentration}, we still utilise Result~\ref{result_3:maxIncreaseQubits} to develop a multi-qubit resource amplification protocol. Specifically, we consider sequentially applying the optimal qubit concentration unitaries detailed in Result~\ref{result_3:maxIncreaseQubits}, feeding the output of two optimal concentration processes into the input of the next (see Appendix C for a schematic example). 

We equate such a concatenation protocol to a greedy optimisation algorithm: the locally optimal choice is always performed, potentially at the expense of reaching the global optimum. Performing a global operation on all available copies of $\rho$ simultaneously will be at least as powerful, and likely more efficient in terms of the local resource increase per number of copies of $\rho$ used. However, the approach presented here offers a fully constructive method for generating a multi-copy coherence amplification process, with all performed unitaries being effective qubit unitaries acting only with in the degenerate subspace of $L_{AB}$ \footnote{Treating subspaces as virtual qubits was shown to be an effective strategy for simplifying the description of dynamics in \cite{PhysRevE.85.051117}}.   

Assuming the input qubits are characterised by $p_{00}^{\recsuper{0}}$ and $p_{01}^{\recsuper{0}}$, after applying $m$ steps of this concatenation protocol, the output qubit state of the $(m+1)$th step will be characterised by
\begin{equation}
    \begin{split} \label{eq:recurrence relations}
        p_{00}^{\recsuper{m+1}} &= p_{00}^{\recsuper{m}} - \frac{2(2p^{\recsuper{m}}_{00}-1) \vert p^{\recsuper{m}}_{01} \vert ^2}{1+(2p^{\recsuper{m}}_{00}-1)^2}; \\
        p_{01}^{\recsuper{m+1}} &=  \vert p^{\recsuper{m}}_{01} \vert ~ \sqrt{1 + (2p^{\recsuper{m}}_{00} - 1)^2}, 
    \end{split}
\end{equation}
which are coupled, non-linear recurrence relations. See SM.~\ref{section: Recurrence Relations Formation} for a proof. Note, after the $m$th step, $2^m$ copies of $\rho$ will have been consumed. Hence, each additional layer leads to an exponential increase in the number of consumed states.

As a first step, we find the fixed points of these recurrence relations: values that remain unchanged when the recurrence is applied (i.e., $p_{00}^{\recsuper{m+1}} = p_{00}^{\recsuper{m}}$ and $p_{01}^{\recsuper{m+1}} = p_{01}^{\recsuper{m}}$). Physically, they represent qubit states for which no coherence can be concentrated, with $p_{00}^{\recsuper{m}}=1/2$ (part of the condition for bound coherence) and $p^{\recsuper{m}}_{01}=0$ (no initial coherence) both found to be fixed points. 

\begin{figure}
    \hspace*{-0.4cm}
    \includegraphics[scale=0.37]{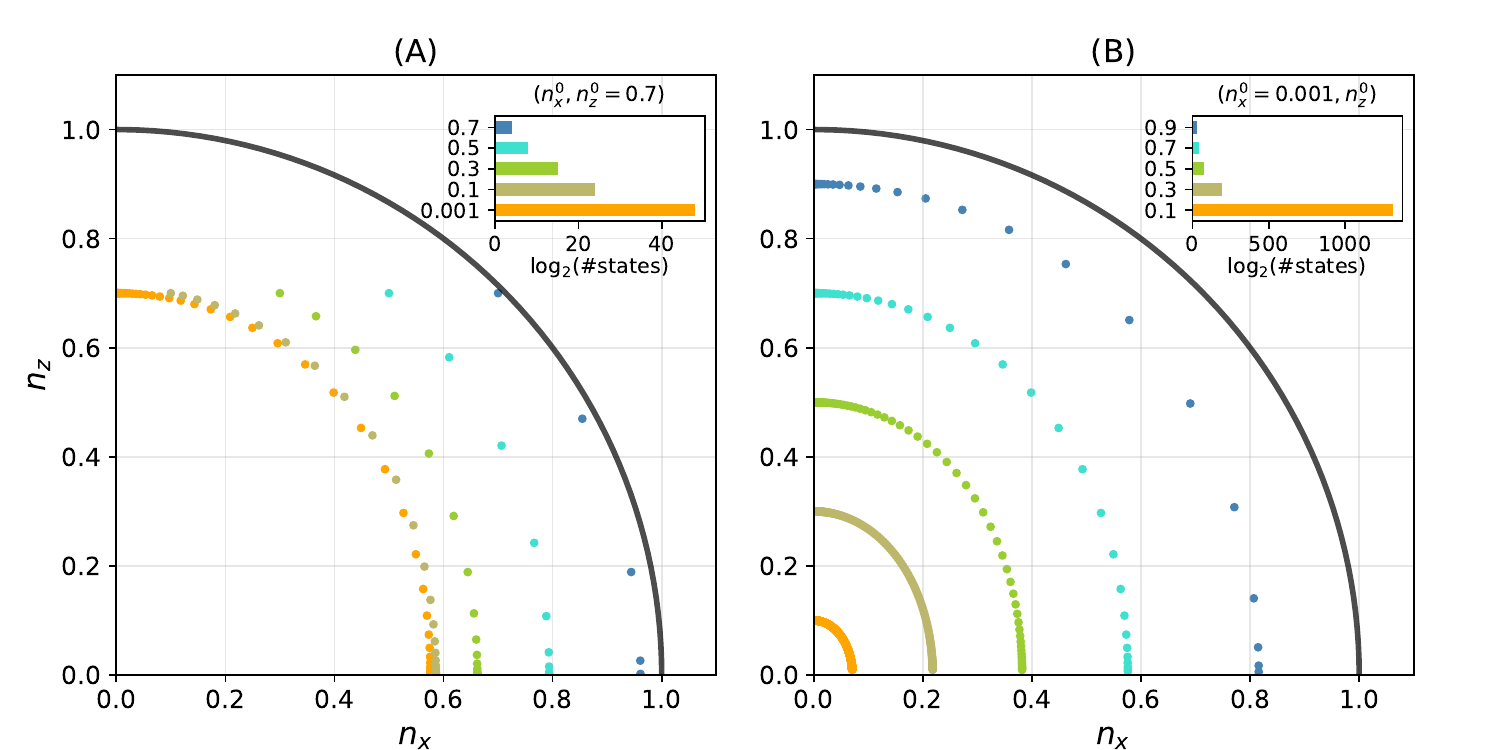}
    \caption{{\bf Qubit concatenation paths through the Bloch sphere.} (A) Paths for $n^0_z=0.7$ and various $n_x^0$. (B) Paths for $n^0_x=0.001$ and various $n_z^0$. Inserts show the $\log_2(\# {\rm ~input~ states})$---the number of steps---needed to converge to within $0.001$ of $n_z=0$, which dictates the maximum reachable coherence.}
    \label{fig:bloch_sphere_plots}
\end{figure}

By setting $p_{00} =(1+n_z)/2$ and $p_{01}=(n_x + in_y)/2$, where $\bm{n} = (n_x, n_y, n_z)$ is the Bloch vector, we plot the output states from each step on the Bloch sphere in Fig.~\ref{fig:bloch_sphere_plots}. This is done for multiple different input states. As $Z$-rotations are locally allowed, we can set $n_y=0$ without loss of generality, leaving $M^{\modesuper{1}}(\rho) = \vert n_x \vert$ to completely characterise the coherence of $\rho$. Since the $Z$-axis and $X$-axis are fixed points of the recurrence relations, we also only need to consider one quarter of the $n_y=0$ plane of the Bloch sphere, as the same behaviour will be replicated in the other quarters.   

It can be seen in Fig.~\ref{fig:bloch_sphere_plots} that even when given access to infinity many copies of $\rho$, the concatenation protocol still cannot reach a maximally coherent state ($n_x=1$) if $\rho$ is not a pure state. 
In fact, the maximum increase in coherence that can be achieved as $m \rightarrow \infty$ can be shown to depend on the purity of $\rho$:
\begin{theorem}\label{lemma:infinte_concat}
    For an initial qubit state $\rho$, let $\sigma_{(m)}$ be the qubit state output from the concatenation protocol's $m$th step.
    Then
    $
        M^{\recsuper{(1)}}(\sigma_{(m)}) \leq \sqrt{2 {\rm tr} \big(\rho^2\big) - 1}
        $ $
    $
\end{theorem}
A proof sketch is given in Appendix D.
The above analysis reveals an interesting feature of coherence amplification.
From Eq.~\eqref{eq:recurrence relations}, and exemplified in Fig.~\ref{fig:bloch_sphere_plots}~(B), even for an initial qubit state with an {\em arbitrarily small amount} of coherence, coherence can still be amplified via the concatenation protocol. Unlike in the asymptotic and catalytic case \cite{PhysRevLett.132.180202, PhysRevA.103.022403}, this is not arbitrary amplification, as the achievable coherence is bounded by Result~\ref{lemma:infinte_concat}. However, there does exist states for which the {{\em ratio}} of the final to initial coherence is {\em unbounded}. In SM.~\ref{appendix:result:exp amplification_proof}, we prove that: 
\begin{theorem}\label{result:exp amplification}
For any {number of states $n=2^N$ where $N \in \mathbb{N}$} and $\epsilon>0$, there exists an initial state $\rho$ {(dependent on $n,\epsilon$)} such that
{
\begin{align}
\hspace{-0.2cm}
M^{\modesuper{1}}(\rho)<\frac{1}{n}\;\&\;
\frac{M^{\modesuper{1}}(\sigma_{(m)})}{M^{\modesuper{1}}(\rho)}>2^{-\epsilon} \sqrt{n}\;\forall\,m\ge \log_{2}(n).
\end{align} 
}
\end{theorem}
Therefore, if $n$ is very large, then one has a large number of initial states with low initial coherence. Result~\ref{result:exp amplification} then states that $\log_2(n)$ number of steps of the concatenation protocol is sufficient to amplify the ratio of the coherence $n^{3/2}$ times. As $n$ can be arbitrarily large, the ratio of initial to final coherence can be unbounded when given access to arbitrarily many copies of the initial state. This finding is complementary to those reported in Ref.~\cite{PhysRevLett.132.180202, PhysRevA.103.022403}, where arbitrary amplifications of unspeakable coherence with unbounded rates are demonstrated in the asymptotic and catalytic transformation. Here, we show that, rather surprisingly, in a very limited setting where only single-qubit states and effective qubit unitaries are used, an unbounded input-output ratio can already be found.

\section{Coherence Concentration Beyond Qubits: Fundamental Limitations}
We now investigate the general concentration cases, considering states with arbitrary finite local dimension.
Given that modes transform independently under allowed operations, in conjugation with Lemma~\ref{global_Modes_to_local_modes}, one only needs to consider the $j$th global mode to concentrate coherence on the $j$th local mode. However, not all components of the global mode will contribute to the local mode. Specifically, one can separate $\mathcal{B}^{\modesuper{j}}_{AB}$ into two subspaces,
\begin{equation}
    \begin{split}
         \mathcal{V}_{\rm in}^{\modesuper{j}} &\coloneqq  {\rm span}\{ \ketbratwo{n+j, m}{n,m}
         \}_{n,m} \\
          \mathcal{V}_{\rm out}^{\modesuper{j}} &\coloneqq  {\rm span}\{ \ketbratwo{n+k, m+j-k}{n,m}
          \}_{n,m,k \neq j},
    \end{split}
\end{equation}
where $\mathcal{B}^{\modesuper{j}}_{AB} = \mathcal{V}^{\modesuper{j}}_{\rm in} \cup \mathcal{V}^{\modesuper{j}}_{\rm out}$ and $\mathcal{V}^{\modesuper{j}}_{\rm in} \cap \mathcal{V}^{\modesuper{j}}_{\rm out}=\emptyset$, such that the basis operators in $\mathcal{V}^{\modesuper{j}}_{\rm in}$ contribute to the local mode, whereas those in $\mathcal{V}_{\rm out}$ do not. The aim of mode concentration can therefore be loosely summarised as maximising the coefficients of the basis operators in $\mathcal{V}^{\modesuper{j}}_{\rm in}$ under allowed unitaries. It is noted that the projector onto $\mathcal{V}_{\rm in}^{\modesuper{j}}$ reads $\Pi_A^{\modesuper{j}} \otimes \Pi_B^{\modesuper{0}}$, where $\Pi^{\modesuper{k}}_{X}$ is the projector onto the $j$th local mode ($X=A,B$). 
Using this, we find the following upper bound on the ability to concentrate coherence in the $j$th local mode. 
\begin{theorem} \label{result:bound_1} Let $d^{(j)}\coloneqq{\rm dim}(\mathcal{V}^{\modesuper{j}}_{\rm in})$.
For any state $\rho$, we have
    \begin{equation}
           \Delta M_{A}^{\scriptscriptstyle(j)} \leq \vert \vert (\rho_A \otimes \rho_B)^{\scriptscriptstyle(j)} \vert \vert_{\text{${d}^{(j)}$}-{\rm KF}} - \vert \vert \rho^{\scriptscriptstyle(j)} \vert \vert_1.
    \end{equation}
\end{theorem}
See Appendix E for a proof sketch. Here, $\| P \|_{\text{$\alpha$-\rm KF}}$ is the $\alpha$-{\em Ky-Fan norm}~\cite{Fan1951}, which is the sum of the $\alpha$ largest singular values of the operator $P$. Hence, as for informational non-equilibrium concentration \cite{hsieh2024informationalnonequilibriumconcentration}, the Ky-Fan norm emerges as a tool for quantifying how much of the global resource content of a state can be relocated into a subsystem. 

In Appendix F, we present an alternative (and more technical) bound on coherence concentration, which is sometimes tighter than the one reported in Result~\ref{result:bound_1} (see also SM.~\ref{appendix: bound comparision} for details). Moreover, see SM.~\ref{Section: qutrit example sup mat} for a qutrit example of both bounds. 

\section{The Role of Initial Correlations and No-Go Results for Coherence Concentration}
Finally, we expand the concentration problem beyond two initially uncorrelated copies, and ask when initial correlations can aid in local coherence enhancement, and when they cannot. 
In $AB$ with equal local dimension $d$, consider a state $
    \rho_{AB} = \rho_{AB}^{(0)} + \sum_{i \geq d}\rho_{AB}^{(i)}
$
such that there exists at least one $k\geq d$ where $\rho_{AB}^{\modesuper{k}} \neq 0$ and $\rho_{AB}^{\modesuper{k}}=0~\forall~ k \in [1,d-1]$. In SM.~\ref{appendix: A Sufficient Condition for The existence of Correlations}, we show that such a mode distribution implies correlations exist between $A$ and $B$. Now, applying any allowed operation to this state will map it to an output state supported on the same modes. Then, using Lemma~\ref{global_Modes_to_local_modes}, it can be concluded that both local states are supported only on the $0$th mode (as the partial trace over modes with $k\geq d$ is zero, making no contribution to the reduced states). Hence, both local states are always symmetric, regardless of the applied allowed operation. Therefore, despite the global state having coherence between the eigenspaces of $L_{AB}$, there exists no allowed operation---not just no allowed unitary---that can increase the local coherence. Furthermore, this argument is {\em independent} of the choice of coherence measure. We thus summarise this as the following no-go result.

\begin{theorem}
    If $[1, d-1] \cap {\rm modes}(\rho_{AB}) = {\{0\}}$, then there exists no allowed operation to concentrate coherence with respect to any coherence measure. 
\end{theorem}

As an explicit example, consider the two-qubit isotropic state, 
$
    \rho_{AB}^{\rm Iso} = p \ketbra{\Phi^+}_{AB} + (1-p)\mathbb{I}/2,
$
where $0 \leq p \leq 1$ and $\ket{\Phi^+}_{AB} = (\ket{00} + \ket{11})_{AB}/\sqrt{2}$. As this state has support only on the $0$th and $2$nd bipartite modes, no coherence can be locally concentrated for all $p>0$, despite there existing correlation between $A$ and $B$ for all $p>0$.
This also means global entanglement---even maximal entanglement---may not always be exchangeable for local coherence, despite the ability to perform global allowed operations.  

\section{Discussion}
Here, we have studied the concentration problem for (unspeakable) coherence. A complete characterisation was found for two qubits, and upper bounds were given beyond qubits. 
From this, we detailed a completely constructive, multi-qubit coherence amplification protocol that uses only effective qubit unitaries involved. {It was shown to be able to unboundedly amplify the ratio of the input-output coherence for certain initial states---a finding that is complementary to the recently reported unbounded coherence amplification in the asymptotic and catalytic regime~\cite{PhysRevLett.132.180202}.} It would be instructive to quantitatively understand how far from the global optimal this protocol is, as it would inform whether the increased practicality of the concatenation protocol is worth the reduction in coherence concentrated. In addition, it would be interesting to introduce finitely many incoherent states into the concentration protocol, and assess whether more coherence enhancement can then be achieved given access to these incoherent states.

Finally, a family of states with bound-coherence was shown to exist, i.e., states that are not maximally coherent, nor incoherent, but for which no allowed unitary exists to perform coherence concentration. We leave it open to determine if such states exist in higher dimensions.

\section{Acknowledgement}
We thank Xueyuan Hu, Valerio Scarani, and Peter Sidajaya for fruitful discussions and comments. B.S.~acknowledges support from UK EPSRC (EP/SO23607/1). C.-Y.H.~acknowledges support from the Royal Society through Enhanced Research Expenses (NFQI) and the Leverhulme Trust Early Career Fellowship (``Quantum complementarity: a novel resource for quantum science and technologies'' with Grant No.~ECF-2024-310). PS Acknowledges support from a CIFAR Azrieli Global Scholarship.

\section*{Appendix}
\subsection{Appendix A: Proof of Lemma~\ref{global_Modes_to_local_modes}}
    The $j$th mode of the bipartite state $\sigma_{AB}$ can be written as
    \begin{equation}
        \sigma_{AB}^{\scriptscriptstyle(j)} = \sum_{k} \sum_{n,m} q_{n+k,m+j-k,n,m} \ketbratwo{n+k, m+j-k}{n,m}_{AB},
    \end{equation}
    which is an operator supported in $\mathcal{B}_{AB}^{(j)}$. Taking the partial trace over the $B$ system of this operator gives
    \begin{equation}
        \begin{split}
            {\rm tr}_B\big(\sigma_{AB}^{\scriptscriptstyle(j)}\big) &= \sum_{n} \Tilde{q}_{n+j,n} \ketbratwo{n+j}{n}_{A} \in \mathcal{B}_A^{\modesuper{j}},    
        \end{split}
    \end{equation}
    where $\Tilde{q}_{n+j,n} = \sum_m q_{n+j,m,n,m}$ and $\mathcal{B}_A^{\modesuper{j}}$ is the $j$th mode subspace of the $A$ system. From the definition of $\Tilde{q}_{n+j,n}$, it can be seen that the matrix elements associated to each basis state are those one would expect from taking the partial trace of the whole system, \hbox{$\Tilde{q}_{n+j,n} = \bra{n+j}  {\rm tr}_B\big(\sigma_{AB}\big) \ket{n}$}. As the operator exists only in $\mathcal{B}^{\modesuper{j}}_A$ and has the correct matrix elements, it is $\sigma_{A}^{\scriptscriptstyle(j)}$.  
\hfill$\square$

\subsection{Appendix B: Proof Sketch of Result~\ref{result_3:maxIncreaseQubits}}
In the two-qubit case, an allowed unitary applies an arbitrary qubit unitary in the $\big\{ \ket{01}, \ket{10} \big\}$ degenerate subspace. This unitary can be decomposed into two $Z$-rotations (with respect to $\ket{01}$ and $\ket{10}$), a $Y$-rotation, and some global phase. The $Z$-rotations and global phase can be shown to only decrease local coherence, leaving only the $Y$-rotation. One can apply this rotation to two copies $\rho$, and then use a geometric argument to find the optimal rotation angle. See SM.~\ref{proof section: result 1} for a full proof. \hfill$\square$

\begin{figure}[h!t]
    \centering
    \includegraphics[width=0.8\linewidth]{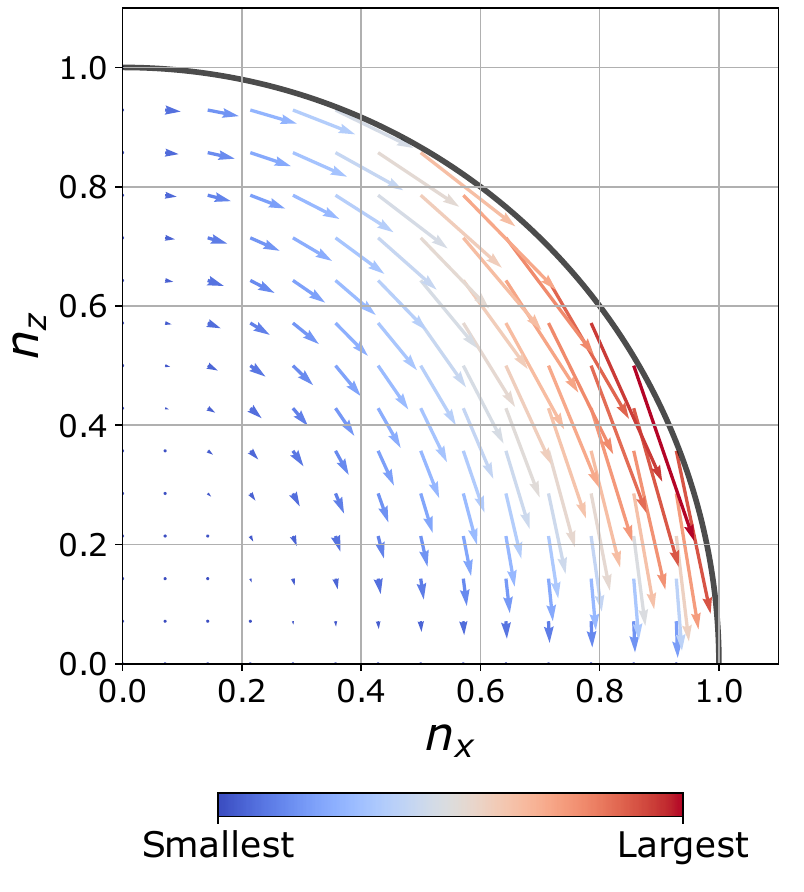}
    \caption{
    {\bf Vector field for the recurrence relations in Eq.~\eqref{eq:recurrence relations}.} The length of the vector and its colour detail how much a given point will change under the recurrence relation. For example, a short blue arrow means a small change will occur, a long red arrow means a large change will occur.
    }
    \label{fig:vector_field}
\end{figure}

\subsection{Appendix C: Details of Concatenation Protocol}
Here, a schematic example for a two-layer qubit concatenation protocol is presented:
\begin{center}
\begin{quantikz}
    \lstick{$\rho_A$} & \gate[2]{{U_{\rm opt}}(\rho)}                  &\ground{} \\
    \lstick{$\rho_B$} &                                        & \gate[2]{{U_{\rm opt}}(\sigma)} & \rstick{$\omega$}\\
    \lstick{$\rho_C$} & \gate[2]{{U_{\rm opt}}(\rho)} &                 &\ground{} \\
    \lstick{$\rho_D$} &                        & \ground{}
\end{quantikz}
\begin{tikzpicture}[overlay, remember picture]
  \draw[red, dashed, thick] 
    node[above, black] {$\sigma$}
    -- ++(0,-1.4cm);
\end{tikzpicture}
\end{center}
The protocol takes four copies of the state $\rho$ as input. A two qubit unitary, denoted by ${U_{\rm opt}}(\rho)$, is applied to the first and second copy, in spaces $A$ and $B$, and the third and fourth copies in spaces $C$ and $D$. Here, ${U_{\rm opt}}(\rho)$ is the optimal concentrating unitary for $\rho$; i.e., if $p_{00}^0 = \bra{0} \rho \ket{0}$ and $p_{01}^0 = \bra{0} \rho \ket{1}$, then 
\begin{equation}
    {U_{\rm opt}}(\rho) = \begin{bmatrix}
        1 & 0 & 0 & 0 \\
        0 & \cos{\theta_{\rm opt}} & -\sin{\theta_{\rm opt}} & 0 \\
        0 & \sin{\theta_{\rm opt}} & \cos{\theta_{\rm opt}} & 0 \\
        0 & 0 & 0 & 1
    \end{bmatrix}, \label{rotationDegSpace}
\end{equation}
where ${\theta_{\rm opt}}=\cos^{-1}\left(1/\sqrt{1+(2p_{00}-1)^2}\right)$.
Let \hbox{$\sigma_B=\textrm{tr}_A\big[{U_{\rm opt}}(\rho) (\rho_A \otimes \rho_B){U_{\rm opt}(\rho)^{\dagger}} \big]$} be the output state into space $B$ from the $1$st concentration step, which is also the output into space $C$. 
Let $p_{00}^1 = \bra{0} \sigma \ket{0}$ and $p_{01}^1 = \bra{0} \sigma \ket{1}$, which can be explicitly calculated via the recurrence relations in Eq.~\eqref{eq:recurrence relations}. From here, the optimal unitary for performing coherence concentration given two copies of $\sigma$, denoted by ${U_{\rm opt}}(\sigma)$, can be found and is of the same form as above. Applying ${U_{\rm opt}}(\sigma)$ to the states in wires $B$ and $C$ will output an even more coherent state $\hbox{$\omega_B=\textrm{tr}_A\big[{U_{\rm opt}}(\sigma) (\sigma_B \otimes \sigma_C){U_{\rm opt}(\sigma)^{\dagger}} \big]$}$. One can repeat this process an arbitrary number of times.

In Fig.~\ref{fig:vector_field}, we plot the vector field of the recurrence relations in Eq.~\eqref{eq:recurrence relations}. The increase in coherence in each step is smallest near the fixed points (along the $x$-axis and $z$-axis), and it can be seen that a smaller $n_x$ (closer to the $z$-axis) leads to smaller step sizes than a smaller $n_z$ (closer to the $x$-axis). 
It can also be seen in Fig.~\ref{fig:vector_field} that the change in $n_z$ and $n_x$ in a single step is highly dependent on the input state, with $n_x$ or $n_z$ close to the fixed points leading to smaller steps. When close to the $n_z$ axis, such that $n_x$ is close to zero, the small step size results from both input states into the concentration having low coherence, and hence there not being much coherence to relocate from system $B$ into system $A$. When close to the $n_x$ axis, however, the step sizes remain small, despite both input states having a (potentially) large amount of coherence. The small step size here instead arises from the fact that increasing coherence when the Bloch vector points predominantly along the $X$-axis is equivalent to increasing the purity of the state. Since it is impossible to increase the local purity given two qubits and unitary operations \cite{hsieh2024informationalnonequilibriumconcentration}, each step size is the small change in coherence that can be managed whilst not increasing the purity.  

\subsection{Appendix D: Proof Sketch of Result~\ref{lemma:infinte_concat}}
It can be shown that under the optimal concentrating unitary, the purity of the state in the $A$ subsystem is reduced. Hence, under each step of the concatenation protocol, the purity is reduced. 
Relating the Bloch vector to purity via its 2-norm and using the formula of $M^{(1)}(\rho)$ completes the proof. See SM.~\ref{appendix:lemma_2_proof} for a full proof. \hfill$\square$

\subsection{Appendix E: Proof Sketch of Result~\ref{result:bound_1}}
The figure of merit $M_A^{\scriptscriptstyle(j)}(\sigma_A)$ can be seen to be upper bounded by the sum of the absolute values of the matrix elements of $\sigma_{AB}^{\scriptscriptstyle(j)}$ associated to basis operators in $\mathcal{V}_{\rm in}^{\modesuper{j}}$. All these matrix elements sit along a diagonal of length $\Tilde{d}$. It can separately be shown that for any operator, the sum of the absolute values of the matrix elements along any diagonal of length $k$ is upper-bounded by the $k$th Ky-Fan norm of the operator. Using the unitary invariance of the Ky-Fan norm and taking away $M^{\scriptscriptstyle(j)}(\rho_A)$ completes the proof. See SM.~\ref{appendix:result_1_proof} for a full proof.
\hfill$\square$

\subsection{Appendix F: Separating Mode Evolution and Bounds On Coherence Concentration}
Here, we present an alternative bound on coherence concentration. We achieve this by noting that a global mode can be separated further into independently transforming components. 
Within the $j$th mode, we introduce the \hbox{$(\lambda_c+j,\lambda_c)$} {\em left-right-degenerate} (LR-D) operator of the mode as $\Pi_{c+j} (\rho \otimes \rho)^{\modesuper{j}} \Pi_{c}$, where $\Pi_{c}$ is the projection onto the eigenspace of $L_{AB}$ with eigenvalue $\lambda_c$ (i.e., $L_{AB} = \sum_c\lambda_c\Pi_c$ with spectrum $\{\lambda_c\}_c$). Under any bipartite allowed unitary $V_{AB}$, the $(\lambda_c+j, \lambda_c)$ LR-D operator of the $j$th mode then evolves independently as
$
    V_{c+j} \big( \Pi_{c+j} (\rho \otimes \rho)^{\modesuper{j}} \Pi_{c} \big) V_{c},
$
where $V_{c}$ is the unitary that acts within the $c$th degenerate subspace of $L_{AB}$. 
One can then perform a similar analysis to Result~\ref{result:bound_1} and define the $(\lambda_c+j, \lambda_c)$th LR-D subspace, $\mathcal{B}^{\modesuper{j,c}}_{AB}$, and  
\begin{equation}
    \begin{split}
         \mathcal{V}_{\rm in}^{\modesuper{j,c}} &\coloneqq  
         {\rm span}\{ \ketbratwo{n+j, m}{n,m}
         \}_{n+m=c } \\
         \mathcal{V}_{\rm out}^{\modesuper{j,c}} &\coloneqq  
         {\rm span}\{ \ketbratwo{n+k, m+j-k}{n,m}
         \}_{n+m=c,~k \neq j}
    \end{split}
\end{equation}
with $\mathcal{B}^{\modesuper{j,c}}_{AB} = \mathcal{V}_{\rm in}^{\modesuper{j,c}} \cup \mathcal{V}_{\rm out}^{\modesuper{j,c}}$ and $\mathcal{V}_{\rm in}^{\modesuper{j,c}} \cap \mathcal{V}_{\rm out}^{\modesuper{j,c}}=\emptyset$. 
The aim of mode concentration is again to maximise the coefficients associated with the basis operators in $\mathcal{V}_{\rm in}^{\modesuper{j,c}}$.
Applying this approach leads to the following upper bound:
\begin{theorem} \label{result:bound_2} 
Let $d^{(j,c)} \coloneqq {\rm dim}(\mathcal{V}_{\rm in}^{\modesuper{j,c}})$.
For any state $\rho$, we have
\begin{equation}
     \Delta M_{A}^{\scriptscriptstyle(j)} \leq \sum_{c=0}^{2d-2-j} \| \Pi_{c+j} (\rho \otimes \rho)^{\scriptscriptstyle(j)} \Pi_{c} \|_{\text{$d^{(j,c)}$}-{\rm KF}} - \onenorm{\rho^{\modesuper{j}}}.
\end{equation}
\end{theorem}
The proof technique is the same as for Result~\ref{result:bound_1}, but now applied individually to each LR-D subspace. See SM.~\ref{appendix:result_2_proof} for the full proof.
Interestingly, in SM.~\ref{appendix: bound comparision}, 
we show that bounds from Results~\ref{result:bound_1} and~\ref{result:bound_2} are incomparable---each provides a tighter upper bound in a variety of cases considered. 

\bibliographystyle{apsrev4-1}
\bibliography{mainTextBib}

\clearpage

\onecolumngrid
\begin{center}
    \Large \bfseries {Supplementary Material: Unspeakable Coherence Concentration}
\end{center}
\vspace{1em}

\let\addcontentsline\oldaddcontentsline %

\begingroup
\parskip=0pt
\setcounter{tocdepth}{2}
\tableofcontents
\endgroup


\section{Relationship between Local and Global Symmetry\label{appendix:globalVsLocal}}

A few observations can be made about the relationship between the local and global (bipartite) symmetry and, hence, the local and global coherence. Firstly, we prove the following Lemma.  
\begin{lemma}[Global Symmetry Implies Local Symmetry]\label{lemma: global vs local symmetry}
Consider a finite-dimensional bipartite system $AB$.
    Let $\rho_{AB}$ be a bipartite state in $AB$ and $L_A$ and $ L_B$ be Hermitian operators acting on systems $A$ and $B$, respectively. Then 
    \begin{equation}
        [\rho_{AB}, L_{AB}] = 0 \implies [\rho_A, L_A] = 0,
    \end{equation}
    where $\rho_A = {\rm tr}_B\big( \rho_{AB} \big) $ and $L_{AB} = L_A \otimes \mathbb{I}_B + \mathbb{I}_A \otimes L_B$, where $L_A, L_B$ are non-degenerate, truncated number operators.  
\end{lemma}

\begin{proof}
    As $L_X$ is a Hermitian (linear) operator ($X=A,B$), its eigenstates form an orthonormal basis, such that $L_X$ can be written as \hbox{$L_X=\sum_{n} \lambda_n^{(X)} \ketbra{n}_X$,} where \hbox{$\braket{n|m} = \delta_{nm}$} $\forall\,n,m$. The bipartite operator $L_{AB}$ can then be written in the joint eigenbasis as 
    \begin{equation}
        L_{AB} = \sum_{i,j} \left(\lambda_n^{(A)} + \lambda_m^{(B)}\right) \ketbra{n}_A \otimes \ketbra{m}_B.
    \end{equation}
    As $[\rho_{AB}, L_{AB}]=0$, $\rho_{AB}$ is block diagonal with respect to the eigenspaces of $L_{AB}$, namely,
    \begin{equation}
        \rho_{AB} = \sum_{n,m,s,t} p_{nmst} \ketbratwo{n}{m}_A \otimes \ketbratwo{s}{t}_B, ~ ~ {\rm where } ~ ~ p_{nmst} > 0 ~ ~ \text{ if and only if} ~ ~ \lambda_n + \lambda_s = \lambda_m + \lambda_t,
    \end{equation}
    and $p_{nmst} = \bra{ns} \rho_{AB} \ket{mt}$. The reduced state in system $A$ is then 
    \begin{equation}
        \rho_A = { \rm tr}_B\big( \rho_{AB} \big) = \sum_{n,m,t} p_{nmtt} \ketbratwo{n}{m}_A.
    \end{equation}
    From here, it can be seen that $p_{nmtt}>0$ if and only if $\lambda_n = \lambda_m$. 
    In particular, this already implies that $\rho_A$ is block diagonal with respect to the eigenbasis of $L_A$, meaning that $[\rho_A, L_A]=0$ and completing the proof.
\end{proof}

Lemma~\ref{lemma: global vs local symmetry} means that global symmetry implies local symmetry.
However, the reverse is not true: the presence of local symmetry does not imply global symmetry.
To see this, consider \hbox{$\rho_{AB} = \ketbra{\Phi^+}_{AB}$} with \hbox{$\ket{\Phi^+}_{AB} = (1/\sqrt{2})(\ket{00}+\ket{11})_{AB}$}, and \hbox{$L_X=\lambda_0\ketbra{0}_X + \lambda_1 \ketbra{1}_X$} for $\lambda_0 \neq \lambda_1$ ($X=A,B$). Locally both systems are maximally mixed and hence symmetric, while the global state is not symmetric with respect to $L_{AB}$. 
In the specific case that $\rho_{AB}$ is a product state, $\rho_{AB}=\rho_A \otimes \rho_B$, then local symmetry in both subsystem ($[\rho_A, L_A]=0$ and $[\rho_B, L_B]=0$) does imply global symmetry ($[\rho_A \otimes \rho_B, L_{AB}]=0$). This can then be extended to hold for any convex mixture of product states that are locally symmetric.
Using these ideas, we now argue that one cannot exploit the differing global and local symmetries in the concentration scenario to increase the local coherence unless there is some initial local coherence present. If $\rho_{AB}$ is a product state and both local systems contain no coherence, then $\rho_{AB}$ is incoherent. Since $V_{AB}$ cannot increase coherence, $V_{AB}(\rho_A \otimes \rho_B)V_{AB}^\dagger$ must also be incoherent. Hence, according to Lemma~\ref{lemma: global vs local symmetry}, the output local systems will be incoherent, too.

{
\section{Proofs of main results \label{appendix:bounds}}
}

\subsection{Proof of Result~\ref{result_3:maxIncreaseQubits}}\label{proof section: result 1}

{We first state Result~\ref{result_3:maxIncreaseQubits} here for clarity.}
\\
\\
\textbf{Result~\ref{result_3:maxIncreaseQubits}}~\textit{Given two copies of a qubit $\rho \in \mathcal{H}^2$, then 
\begin{equation}
    \max_{V_{AB}} ~ \Delta M_A^{\modesuper{1}} = \vert p_{01} \vert ~ \bigg( \sqrt{1 + (2p_{00} - 1)^2} - 1 \bigg),
\end{equation}
where the maximum is over all resource non-increasing unitaries, with the optimal achieved by applying a rotation of angle $(\sqrt{1+(2p_{00}-1)^2})^{-1}$ in the degenerate subspace. }

\begin{proof}
    Firstly, we note that an arbitrary single-qubit state $\rho$ can be written as
    \begin{equation}
        \rho = \begin{bmatrix}
            p_{00} & p_{01} \\
            p_{01}^* & 1-p_{00}
        \end{bmatrix},
    \end{equation}
    where the decomposition is with respect to the eigenbasis of $L$, and $p_{00}=\bra{0} \rho \ket{0}, p_{01}=\bra{0} \rho \ket{1}$.  
    To two copies of $\rho$, we apply an allowed unitary $V_{AB}$ on it (i.e., $[V_{AB}, L_{AB}] = 0$). The most general form of such a unitary is 
    \begin{equation}
         V_{AB} = \begin{bmatrix}
        e^{i \omega_0} & 0 & 0 & 0 \\
        0 & U_{00} & U_{01} & 0 \\
        0 & U_{10} & U_{11} & 0 \\
        0 & 0 & 0 & e^{i \omega_1}
    \end{bmatrix}, 
    \end{equation}
    where $0 \leq \omega_i \leq 2 \pi$ for $i \in \{0,1\}$ and
    \begin{equation}
       \begin{split}
{       \begin{bmatrix}
            U_{00} & U_{01} \\
            U_{10} & U_{11}
        \end{bmatrix} }
        &= e^{-i \alpha} \begin{bmatrix}
            e^{-i \phi_1 } & 0 \\
            0 & e^{i \phi_1}
        \end{bmatrix}
{        \begin{bmatrix}
            \cos\theta & -\sin\theta \\
            \sin\theta & \cos{\theta} 
        \end{bmatrix}}
        \begin{bmatrix}
            e^{-i \phi_2 } & 0 \\
            0 & e^{i \phi_2}
        \end{bmatrix}
       \end{split}
    \end{equation}
    is a general two-qubit unitary in the degenerate subspace, which has be decomposed into a combination of two rotations about the $Z$-axis (defined by the states $\{ \ket{01}, \ket{10} \}$), a rotation about the $Y$-axis, and a phase factor, $\alpha$. Hence, the two-qubit unitary $V_{AB}$ can be written as a product of unitaries,   
        \begin{equation}\label{Eq:product of unitaries}
        V_{AB}=S(\omega_0)T(\omega_1)R^{z}(2\phi_1)R^{y}(2\theta)R^{z}(2 \phi_2)
    \end{equation}
    where 
        \begin{equation}
            \begin{split}
                    S(\omega_0) = \begin{bmatrix}
                        e^{i \omega_0} & 0 & 0 & 0 \\
                        0 & 1 & 0 & 0 \\
                        0 & 0 & 1 & 0 \\
                        0 & 0 & 0 & 1
                    \end{bmatrix}, ~  T(\omega_1) &= \begin{bmatrix}
                        1 & 0 & 0 & 0 \\
                        0 & 1 & 0 & 0 \\
                        0 & 0 & 1 & 0 \\
                        0 & 0 & 0 & e^{i \omega_1}
                    \end{bmatrix}, 
                        R^{z}(2 \phi) = \begin{bmatrix}
                        1 & 0 & 0 & 0 \\
                        0 & e^{-i \phi} & 0 & 0 \\
                        0 & 0 & e^{i \phi} & 0 \\
                        0 & 0 & 0 & 1
                    \end{bmatrix}, \\
                        R^{y}(2 \theta) &= \begin{bmatrix}
                        1 & 0 & 0 & 0 \\
                        0 & {\cos\theta} & -\sin{\theta} & 0 \\
                        0 & \sin{\theta} & \cos{\theta} & 0 \\
                        0 & 0 & 0 & 1
                    \end{bmatrix}, 
            \end{split}
        \end{equation}
    with all decompositions with respect to the eigenbasis of $L_{AB}$. 
    
    To achieve coherence concentration, we need to evaluate the output state $\sigma_{A} = {\rm tr}_{B}\big[ V_{AB} (\rho_A \otimes \rho_B) V_{AB}^\dagger \big]$, which is characterised by $q_{00} \coloneqq \bra{0} \sigma_A \ket{0}$ and $q_{01} \coloneqq \bra{0} \sigma_A \ket{1}$. To this end,
    we now sequentially apply each unitary in Eq.~\eqref{Eq:product of unitaries} to $\rho_A \otimes \rho_B$, assess each's ability to increase local coherence.      
    Firstly, if applying $S(\omega_0)$, 
    \begin{equation}
            \vert q_{01} \vert = \vert p_{01} (1-p_{00} + e^{i\omega_0}p_{00}) \vert\leq \vert p_{01} \vert \big( \vert 1 - p_{00} \vert + \vert e^{i \omega_0} p_{00} \vert \big)= \vert p_{01} \vert,
    \end{equation}
    where the equality occurs when $e^{i \omega_0}=1$. Hence, under unitaries of this form, we have $\Delta M^{\modesuper{j}} \leq 0$ $\forall ~ \omega_0$, with equality if $\omega_0 = 2\pi n$ for integers $n$. The same can be shown for $T(\omega_1)$, meaning that when looking to concentrate coherence, it is always sufficient to set $\omega_0 = 0 = \omega_1$ and hence apply the identity in these subspaces. 
    
    Secondly, when applying $R^{z}(2 \phi)$ to $\rho_A \otimes \rho_B$ and taking the partial trace, it can be seen that $\vert q_{01} \vert = \vert p_{01} \vert ~\forall~ \phi$, meaning that $\Delta M^{\modesuper{j}} = 0$ for all unitaries of this form. Hence, when looking to concentrate coherence, it is sufficient to always set $\phi = 0$. It is noted that $[S(\omega_0), R^{z}(2\phi_1)R^{y}(2\theta)R^{z}(2 \phi_1)]=0$ and $[T(\omega_1), R^{z}(2\phi_1)R^{y}(2\theta)R^{z}(2 \phi_1)]=0$, meaning that changing the order in which these unitaries are applied cannot help one concentrate coherence. One therefore only has to consider the local coherence increasing power of $R^{y}(2 \theta)$. 
    
    Applying $R^{y}(2 \theta)$ to $\rho_A \otimes \rho_B$ and tracing out the $B$ system then gives 
        $
        q_{01} = p_{01} \big[ \cos\theta + (2p_{00}-1)\sin{\theta} \big].
        $
    From here, the maximum $M^{\modesuper{j}}(\sigma_A)$ that can achieved {by allowed unitaries} is given by
    \begin{equation}
        \begin{split}
            \max_{V_{AB}} M^{\modesuper{j}}(\sigma_A) = \max_{\theta} {\big\vert p_{01} \big[ \cos\theta + \sin\theta(2p_{00}-1) \big] \big\vert} {= \vert p_{01} \vert ~ \max_{\bm{u}} \vert \bm{u} \cdot \bm{v} \vert,}
        \end{split}
    \end{equation}
    where 
    \begin{equation}
        \bm{u} = {\begin{bmatrix}
            \cos\theta \\
            \sin\theta
        \end{bmatrix}}, ~ ~ \bm{v} = \begin{bmatrix}
            1 \\
            2p_{00}-1
        \end{bmatrix}.
    \end{equation}
    The dot product will be maximised when $\bm{u}$ points in the same direction as $\bm{v}$. Hence, the maximum {is achieved by}
    \begin{equation}
         {\bm{u}_{\rm opt}} = \frac{1}{\sqrt{1 + (2p_{00}-1)^2}}\begin{bmatrix}
            1 \\
            2p_{00}-1
        \end{bmatrix},
    \end{equation}
    from which the optimal angle of rotation to increase the local coherence can be found:
    \begin{align}\label{Eq:optimal angle}
    \theta_{\rm opt}\coloneqq\cos^{-1}\left(\frac{1}{\sqrt{1 + (2p_{00}-1)^2}}\right). 
    \end{align}
    Finally, using $\bm{u}_{\rm opt}$, the maximum {$\Delta M^{\modesuper{j}}_A$} can be found to be  
    \begin{equation}
        \begin{split}
            \Delta M_{A}^{\modesuper{j}}  \coloneqq \max_{V_{AB}} M^{\modesuper{j}}(\sigma_A) -  M^{\modesuper{j}}(\rho_A) {= \vert p_{01} \vert ~ \bigg( \sqrt{1 + (2p_{00}-1)^2} \bigg) - \vert p_{01} \vert,}
        \end{split}
    \end{equation}
    completing the proof. 
\end{proof}

\subsection{Proof of Eq.~\ref{eq:recurrence relations}}\label{section: Recurrence Relations Formation}

\begin{proof}
{Here, we will show} how the coupled non-linear recurrence relationships, {namely, Eq.~\eqref{eq:recurrence relations} in the main text,} are formed from the unitary found in Result~\ref{result_3:maxIncreaseQubits}. 
Assuming one initially has two copies of the state $\rho$:
\begin{equation}
    \begin{split}
        \rho_A \otimes \rho_B  &= \begin{bmatrix}
            p_{00} & p_{01} \\
            p_{01}^* & (1-p_{00})
        \end{bmatrix} \otimes \begin{bmatrix}
            p_{00} & p_{01} \\
            p_{01}^* & (1-p_{00})
        \end{bmatrix} \\
        &= \begin{bmatrix}
            p_{00}^2       & p_{00}p_{01}           &  p_{01}p_{00}             & p_{01}^2            \\
            p_{00}p_{01}^* & p_{00}(1-p_{00})         &  \vert p_{01} \vert^2    & p_{01}(1-p_{00})    \\   
            p_{01}^*p_{00} & \vert p_{01} \vert^2  &  (1-p_{00})p_{00}     & (1-p_{00}) p_{01}   \\   
            (p_{01}^*)^2 & p_{01}^*(1-p_{00})         &  (1-p_{00})p_{01}^*    & (1-p_{00})^2   \\   
        \end{bmatrix}.
    \end{split}
\end{equation}
To this bipartite state, the following optimal concentrating unitary, found in Result~\ref{result_3:maxIncreaseQubits}, is applied, 
\begin{equation}
    {U_{\rm opt}}(\rho) = \begin{bmatrix}
        1 & 0 & 0 & 0 \\
        0 & \cos{\theta_{\rm opt}} & -\sin{\theta_{\rm opt}} & 0 \\
        0 & \sin{\theta_{\rm opt}} & \cos{\theta_{\rm opt}} & 0 \\
        0 & 0 & 0 & 1
    \end{bmatrix}, \label{rotationDegSpace_2}
\end{equation}
{where, from Eq.~\eqref{Eq:optimal angle}, we have}
\begin{equation}
    \cos{\theta_{\rm opt}} = \frac{1}{\sqrt{1+(2p_{00}-1)^2}}, ~ ~ ~ ~ \sin{\theta_{\rm opt}} = \frac{2p_{00}-1}{\sqrt{1+(2p_{00}-1)^2}},
\end{equation}
{which are dependent of $\rho$.} The $B$ system is then traced over, giving an output state 
\begin{equation}
    \sigma_A = \textrm{tr}_B\big[ {U_{\rm opt}}(\rho) (\rho_A \otimes \rho_B ) {U_{\rm opt}(\rho)^{\dagger}} \big] {\eqqcolon \begin{bmatrix}
            q_{00} & q_{01} \\
            q_{01}^* & (1-q_{00})
        \end{bmatrix}.}
\end{equation}
By explicitly performing {this calculation,} it can be seen that 
\begin{equation}
    \begin{split}
        q_{00} &= p_{00}^2 + p_{00}(1-p_{00}) \cos^2{\theta_{\rm opt}} - \vert p_{01} \vert^2 \sin{\theta_{\rm opt}} \cos{\theta_{\rm opt}} - \vert p_{01} \vert^2 \sin{\theta_{\rm opt}} \cos{\theta_{\rm opt}} \\
        &= p_{00} - 2 \vert p_{01} \vert^2 \sin{\theta_{\rm opt}} \cos{\theta_{\rm opt}},
    \end{split}
\end{equation}
whilst 
\begin{equation}
    \begin{split}
        q_{01} &= p_{00}p_{01}\sin{\theta_{\rm opt}} + p_{01}p_{00}\cos{\theta_{\rm opt}} + p_{01}(1-p_{00})\cos{\theta_{\rm opt}} - (1-p_{00})p_{01} \sin{\theta_{\rm opt}} \\
        &= p_{01} \bigg[ \cos{\theta_{\rm opt}} + \sin{\theta_{\rm opt}}(2p_{00}-1) \bigg].
    \end{split}
\end{equation}
Inputting $\cos{\theta_{\rm opt}}$ and $\sin{\theta_{\rm opt}}$ into these equation then gives 
\begin{equation}
    \begin{split}
        q_{00} = p_{00} - \frac{2(2p_{00}-1) \vert p_{01} \vert^2}{1+(2p_{00}-1)^2}; \\
        q_{01} = p_{01} \sqrt{1+(2p_{00}-1)^2}.
    \end{split}
\end{equation}
The output state $\sigma$, as characterised by $q_{00}$ and $q_{01}$, therefore depends only on 
{$p_{00}$ and $p_{01}$ from $\rho$.} We then note that local {$Z$-rotations} are always {allowed}, such that we can set $p_{01} \rightarrow \vert p_{01} \vert$ without loss of generality. 
One can now consider taking two copies of $\sigma$ and repeating the same process. This would output the state 
\begin{equation}
    \omega_A = \textrm{tr}_B\big[ {U_{\rm opt}(\sigma)} (\sigma_A \otimes \sigma_B ) {U_{\rm opt}(\sigma)^{\dagger}} \big] {\eqqcolon} \begin{bmatrix}
            l_{00} & l_{01} \\
            l_{01}^* & (1-l_{00})
        \end{bmatrix},
\end{equation}
where, using the same calculation as above, it can be seen that 
\begin{equation}
    \begin{split}
        l_{00} = q_{00} - \frac{2(2q_{00}-1) \vert q_{01} \vert^2}{1+(2q_{00}-1)^2}; \\
        l_{01} = \vert q_{01} \vert \sqrt{1+(2q_{00}-1)^2}.
    \end{split}
\end{equation}
Again, the output state $\omega$, as characterised by $l_{00}$ and $l_{01}$, can be seen to depend on 
{$q_{00}$ and $q_{01}$ from $\sigma$.} This process can then be repeated an arbitrary {finite} number of times, where each time the parameters characterising the output state can be formulated in terms of the parameters characterising the input state. {This, therefore,} takes the form of recurrence relations, which here are coupled and non-linear functions where the next value of the function depends on some combination of the previous value. 
We note that here the outputs also only depend on the previous values, and not any further back in the sequence. By labelling the input state from the $m$th layer as {$p^{\recsuper{m}}_{00}$ and $p^{\recsuper{m}}_{01}$,} the outputs from the $m+1$ layer can therefore be given by 
\begin{equation}
    \begin{split} \label{SuppMat_eq:recurrence relations}
        p_{00}^{\recsuper{m+1}} &= p_{00}^{\recsuper{m}} - \frac{2(2p^{\recsuper{m}}_{00}-1) \vert p^{\recsuper{m}}_{01} \vert ^2}{1+(2p^{\recsuper{m}}_{00}-1)^2}; \\
        p_{01}^{\recsuper{m+1}} &=  \vert p^{\recsuper{m}}_{01} \vert ~ \sqrt{1 + (2p^{\recsuper{m}}_{00} - 1)^2},
    \end{split}
\end{equation}
{which is exactly Eq.~\eqref{eq:recurrence relations} in the main text.
The proof is thus completed.}
\end{proof}

\subsection{Proof of Result~\ref{lemma:infinte_concat} \label{appendix:lemma_2_proof}}

This appendix contains a proof of Result~\ref{lemma:infinte_concat}, which is first stated here for clarity. 
\\
\\
\textbf{Result~\ref{lemma:infinte_concat}}~\textit{Given an initial qubit state $\rho$, {let $\sigma_{(m)}$ be the qubit state output after the $m$th step of the concatenation protocol. Then we have}
}
    \begin{equation}
        {M^{\modesuper{1}}(\sigma_{m}) \leq \sqrt{2 {\rm tr} \big(\rho^2\big) - 1}\quad\forall\,m.}
    \end{equation}
\begin{proof}
    Firstly, if $\rho$ is already incoherent, we must have $M^{(1)}(\sigma_{(m)})=0$ $\forall\,m$ and the bound holds immediately. Hence, it suffices to focus on $\rho$ with some initial coherence.
    
    Now, let $\bm{n}^0 = (n^0_x, n^0_y, n^0_z) \in \mathbb{R}^3$ be $\rho$'s Bloch vector, which has 2-norm \hbox{$\| \bm{n}^0 \|_2 = \sqrt{(n_x^0)^2 + (n_y^0)^2 +(n_z^0)^2}$}. Let $p^0_{00}\coloneqq\bra{0}\rho\ket{0}=(1+n^0_z)/2$ and $p^0_{01}\coloneqq\bra{0}\rho\ket{1}=(n^0_x + in^0_y)/2$. As stated in the main text, since $Z$-rotations are all locally allowed, we can set $n_y=0$ without loss of generality.
    Also, as we assume $\rho$ is coherent, we must have $n_x^0\neq0$.
    This means $\rho$'s purity, ${\rm tr}(\rho^2)$, can be expressed as
    \begin{align}\label{Eq:purity formula}
    {\rm tr}\big( \rho^2 \big) = (p_{00}^0)^2 + \left(1-p_{00}^0\right)^2 + 2|p_{01}^0|^2 = \frac{1+\| \bm{n}^{0} \|_2^2}{2}.
    \end{align}
    Furthermore, the recurrence relations [Eq.~\eqref{eq:recurrence relations} in the main text] can be rewritten in terms of the Bloch vector as 
    {\begin{equation}\label{Eq: formula001}
        \begin{split}
            n_z^{m+1} &= n_z^m - \frac{n_z^m(n_x^m)^2}{1+(n_z^m)^2}; \\
            |n_x^{m+1}| &= |n_x^m| \sqrt{1+(n_z^m)^2},
        \end{split}
    \end{equation}
    where $n_z^m$ and $n_x^m$} are the $Z$ and $X$ components of {$\sigma_{(m)}$'s Bloch vector, denoted by $\bm{n}^m$. We note that $n_y^m=0$ $\forall\,m$ such that  
    $M^{(1)}(\sigma_{(m)}) = |n_{x}^{\recsuper{m}}|.$
    Now, a direct calculation shows that 
    \begin{equation}\label{eq:computation001}
     \| \bm{n}^{m+1} \|_2^2 = \| \bm{n}^m \|_2^2 +\frac{(n_z^mn_x^m)^2}{1+(n_z^m)^2} \biggl[ (n_z^m)^2+\frac{(n_x^m)^2}{1+(n_z^m)^2}-1 \biggl].
    \end{equation}
    It is then noted that 
    \begin{equation}\label{eq:computation002}
        (n_z^m)^2+\frac{(n_x^m)^2}{1+(n_z^m)^2} = \frac{(n_z^m)^2+(n_x^m)^2+(n_z^m)^4}{1+(n_z^m)^2}\le\frac{1+(n_z^m)^4}{1+(n_z^m)^2}=1-(n_z^m)^2 \leq 1\quad\forall\,m,
    \end{equation}
    as $(n_z^m)^2 + (n_x^m)^2 \leq 1$ due to normalisation and $1+(n_z^m)^4 = [1+(n_z^m)^2][1-(n_z^m)^2]$. As
    $(n_z^mn_x^m)^2/[1+(n_z^m)^2] \geq 0$ $\forall\,m$, Eqs.~\eqref{eq:computation001} and~\eqref{eq:computation002} jointly imply that
    \begin{align}
    \| \bm{n}^{m+1} \|_2 \leq \| \bm{n}^{m} \|_2\quad\forall\,m. 
    \end{align}
    Using Eq.~\eqref{Eq:purity formula}, the definition of the $2$-norm, and recalling that $\sigma_{(m)}$'s coherence reads
    $M^{(1)}(\sigma_{(m)}) = |n_{x}^{\recsuper{m}}|$, we obtain 
    \begin{align}
    M^{(1)}(\sigma_{(m)})=|n_x^m|\le\| \bm{n}^{m} \|_2 \leq \| \bm{n}^{0} \|_2 = \sqrt{2{\rm tr}(\rho^2)-1}\quad\forall\,m,
    \end{align}
    which is the desired result.}
\end{proof}

{
As a remark, we have the following observation:
\begin{center}
{\em If the initial qubit state $\rho$ is coherent, then, no matter how small the initial coherence is, we must have $n_z^\infty=0$.}
\end{center}
Namely, a vanishing amount of coherence possessed by $\rho$ suffices to asymptotically erase the $Z$ component of the Bloch vector.
\begin{proof}
To see this, first, from the above proof, we observe that
    $
    n_z^{m+1} \le n_z^m$ and $|n_x^{m+1}| \ge |n_x^m|$ $\forall\,m\in\mathbb{N},
    $ 
    meaning that $\{|n_x^m|\}_m$ is monotonic non-decreasing and $\{n_z^m\}_m$ is monotonic non-increasing when $m\to\infty$.
    Hence, by the monotone convergence theorem (see, e.g., Ref.~\cite{Apostol}), we must have
    \begin{align}
    \lim_{m\to\infty}n_z^m = \inf_m n_z^m\eqqcolon n_z^\infty\quad\&\quad \lim_{m\to\infty}|n_x^m| = \sup_m |n_x^m|\eqqcolon n_x^\infty.
    \end{align}
    Now we argue that $n_z^\infty=0$. To see this, we apply $m\to\infty$ on both sides of $|n_x^{m+1}| = |n_x^m| \sqrt{1+(n_z^m)^2}$ [i.e., the second line in Eq.~\eqref{Eq: formula001}], obtaining $(n_x^{\infty})^2 = (n_x^\infty)^2\left[1+(n_z^\infty)^2\right]$ (which is true because both limits $\lim_{m\to\infty}n_z^m$ and $\lim_{m\to\infty}|n_x^m|$ exist).
    This then implies that $|n_x^\infty|n_z^\infty=0$.
    Hence, we must have either $n_z^\infty=0$ or $|n_x^\infty|=0$. However, if $n_x^\infty=0$, then, since $\{|n_x^m|\}_m$ is monotonic non-decreasing, we must have $n_x^0=0$, meaning that the initial state $\rho$ is already incoherent, in contradiction to our assumption.
    Consequently, we must have 
    $n_z^\infty=0.$
\end{proof}
}

\subsection{Proof of Result~\ref{result:exp amplification}}\label{appendix:result:exp amplification_proof}

This appendix contains a proof of Result~\ref{result:exp amplification}, which is first stated here for clarity. 
\\
\\
{\textbf{Result~\ref{result:exp amplification}}~\textit{For any number of states $n \in\mathbb{N}$ and $\epsilon>0$, there exists an initial state $\rho$ (dependent on $n,\epsilon$) such that}} {
\begin{align}
\hspace{-0.1cm}
M^{\modesuper{1}}(\rho)<\frac{1}{n}\;\&\;
\frac{M^{\modesuper{1}}(\sigma_{(m)})}{M^{\modesuper{1}}(\rho)}>2^{-\epsilon} \sqrt{n}\;\forall\,m\ge \log_{2}(n).
\end{align} }
\begin{proof}
For a given $N\in\mathbb{N}$, we firstly choose an initial qubit state $\rho$ with Bloch vector $(n_x^0,n_y^0=0,n_z^0)$ satisfying
\begin{align}\label{Eq:conditions}
0<|n_x^0|^2\ll\frac{1}{2^{2N}}\times\frac{1}{2N}\times \min\{ 1-|n_z^0|^2, |n_z^0|^2 \},
\end{align}
where the notation ``$x\ll y$'' means that $x/y\sim O(1)$ is significantly smaller than $1$. As a result, this state has a vanishing amount of coherence, as 
\begin{align}\label{Eq:vanishing initial coherence}
M^{(1)}(\rho) = |n_x^0|\ll\frac{1}{2^N}.
\end{align}
For convenience, let us define, for $m\in\mathbb{N}\cup\{0\}$,
\begin{align}
A_m\coloneqq|n_x^m|^2\quad\&\quad B_m\coloneqq|n_z^m|^2.
\end{align}
Then Eq.~\eqref{Eq:conditions} becomes
\begin{align}\label{Eq:conditions var2}
0<A_0\ll\frac{1}{2^{2N}}\times\frac{1}{2N}\times \min\{1-B_0, B_0 \}.
\end{align}
Also, using Eq.~\eqref{Eq: formula001}, we have that, for every $m\in\mathbb{N}\cup\{0\}$,
\begin{equation}\label{Eq: A_mB_m formula001}
        \begin{split}
            B_{m+1} &= B_m\left(1-\frac{A_m}{1+B_m}\right)^2 = B_0\prod_{k=0}^{m}\left(1-\frac{A_k}{1+B_k}\right)^2; \\
            A_{m+1} &= A_m \left(1+B_m\right) = A_0\prod_{k=0}^m\left(1+B_k\right).
        \end{split}
    \end{equation}
Using the fact that $0\le B_k\le1$ $\forall\,k$, we can combine Eqs.~\eqref{Eq:conditions var2} and~\eqref{Eq: A_mB_m formula001} to obtain
\begin{align}\label{Eq: condition property01}
A_m {\leq A_0 \prod_{k=0}^m (2)} = 2^mA_0\le2^NA_0\ll\frac{1}{2^N}\times\frac{1}{2N}\times \min\{1-B_0, B_0 \} \quad\forall\,m\le N.
\end{align}
Where in the first inequality we have set $B_k=1$, and in the second we have used the fact that $m \leq N$. Finally, we have subbed in Eq.~\eqref{Eq:conditions var2}. 

Furthermore, using the fact that $A_0\le A_k\le1$ and $B_0\ge B_k\ge0$ $\forall\,k$, Eq.~\eqref{Eq: A_mB_m formula001} also implies that, for every $m\le N$ (note, the justification details why we get to the next line),
\begin{align}\label{Eq: condition property02}
    B_0 &\geq B_{m+1} \hspace{5.2cm} \textrm{as}~ 
    \text{Eq.~\eqref{Eq: A_mB_m formula001}}
    \nonumber \\
    &= B_0\prod_{k=0}^{m}\left(1-\frac{A_k}{1+B_k}\right)^2 \hspace{2.6cm} \textrm{as}~B_k \geq 0 ~ \forall~k  \nonumber \\
    & \geq B_0\prod_{k=0}^{m}\left(1-A_k\right)^2 \hspace{3.5cm}   \textrm{as}~ A_m \geq A_{m-1} \nonumber \\
    & \geq B_0\prod_{k=0}^{m}\left(1-A_m\right)^2 = B_0 \left(1-A_m\right)^{2m} \hspace{0.7cm} \textrm{as} ~ m \leq N \\
    & \geq B_0 \left(1-A_m\right)^{2N} \hspace{3.7cm} \textrm{as} ~(1+x)^n = \sum_{k=0}^n {n \choose k} x^k ~ ~ \textrm{ignoring negative on $A_m$}  \nonumber \\
    &\geq B_0\left(1-\sum_{k=1}^{2N}(2NA_m)^k\right) \hspace{1cm} \hspace{1.4cm} \textrm{as sum of a geometric series} \nonumber \\
    &=B_0\left(1-2NA_m\times\frac{1-(2NA_m)^{2N}}{1-2NA_m}\right) \hspace{0.6cm} \textrm{as}~(2NA_m)^{2N}\ge0
    \nonumber \\
    & \geq B_0\left(1-\frac{2NA_m}{1-2NA_m}\right) \hspace{2.7cm} \textrm{as}~0 \leq 2NA_m \ll 2^{-N} < 1/2~\text{due to Eq.~\eqref{Eq: condition property01}}  \nonumber\\
    & \geq  B_0-4NA_m, \nonumber
\end{align}
giving in total $B_{m+1} \geq B_0 -4 N A_m$.
This means that, as long as $m\le N$, we have $B_{m+1} \approx B_0$, due to Eq.~\eqref{Eq: condition property01}.

Now, by setting $m=N-1$, Eqs.~\eqref{Eq: A_mB_m formula001} and~\eqref{Eq: condition property02} jointly suggest that (below, we define $A_{-1}\coloneqq0$, so that $B_0 \geq B_0-4NA_{-1}$)
\begin{align}\label{Eq: computation001 part1}
\log_2\frac{A_{N}}{A_0} &=\sum_{k=0}^{N-1}\log_2(1+B_k) \hspace{2.9cm} \textrm{as}~B_{k+1} \geq B_0-4NA_k\;\forall\,k\le N~\text{from Eq.~\eqref{Eq: condition property02}}  \nonumber\\
& \ge \sum_{k=0}^{N-1}\log_2\left(1+B_0-{4}NA_{k-1}\right) \hspace{1cm}~\textrm{as} ~ A_{k} \le A_{N}~\forall\,k\le N \;\&\;A_{-1}\coloneqq0 \nonumber \\
&\ge \sum_{k=0}^{N-1} \log_2\left(1+B_0-{4}NA_{N}\right) \\
&=  N \log_2\left(1+B_0-{4}NA_{N}\right)   \nonumber \\
&= N\log_2\left[2 - \left(1-B_0+{4}NA_{N}\right)\right] \hspace{0.7cm}~\textrm{as $4NA_N\ll(1-B_0)/2^{N-1}\le1-B_0$ due to Eq.~\eqref{Eq: condition property01}}~\nonumber \\
&\ge N\log_2\left[2 - 2(1-B_0)\right] \nonumber \\
&= N +N\log_2B_0, \nonumber
\end{align}
Since $\{A_k\}_k$ is non-decreasing in $k$, and since $M^{(1)}(\sigma_{(m)}) = |n_x^m| = \sqrt{A_m}$ as well as $B_0=|n_z^0|^2$, we thus conclude that $\forall\,m\ge N$ the following holds
\begin{align}
\frac{M^{(1)(\sigma_{(m)})}}{M^{(1)}(\rho)}&\ge \frac{M^{(1)(\sigma_{(N)})}}{M^{(1)}(\rho)} 
=\sqrt{\frac{A_N}{A_0}} 
\geq\sqrt{2^{N+N\log_2(|n_z^0|^2)}} 
= 2^{\frac{N}{2}}\times2^{N\log_2 (|n_z^0|)},
\end{align}
where we have used Eq.~\eqref{Eq: computation001 part1}.
Note that $\log_2(|n_z^0|)\le0$ is in general negative.
Finally, to see the exponential amplification, for any given $\epsilon>0$, which can be as small as we want, there exists an $|n_z^0|$ satisfying $0<\left|\log_2|n_z^0|\right|\ll\epsilon/N$, from which we have $N\log_2|n_z^0|\ge-\epsilon$. 
Put this together, we have,
\begin{align}
\frac{M^{(1)(\sigma_{(m)})}}{M^{(1)}(\rho)}\ge2^{\frac{N}{2}-\epsilon}\quad\forall\,m\ge N.
\end{align}
{To complete the proof, we note that $n=2^N$, where $n$ is the number of states used in $N$ layers of the protocol.} Hence, there exists an initial state that has a vanishingly small amount of coherence [Eq.~\eqref{Eq:vanishing initial coherence}], from which a state with unboundedly more coherence can be created using this concatenation protocol.  
\end{proof}
Numerically, it can be seen in Fig~\ref{fig:bloch_sphere_plots} that even for states with $0 \leq |n_z^0| < 1/2$, given enough steps, an exponential increase in local coherence can be achieved. We leave a mathematical proof of this fact for future work.

\subsection{Proof of Result~\ref{result:bound_1} \label{appendix:result_1_proof}}

{To prove Result~\ref{result:bound_1}, the following Lemma is first proved.} It states that, given some operator {$P$,} the sum over the absolute value of the matrix elements of any diagonal of length $k$, in any basis, is upper-bounded by the $k$th Ky-Fan norm of {$P$.}

\begin{lemma} \label{lemma_sum_coefficents}
    Let 
    {$P$ be a linear operator acting on finite-dimensional vector spaces, where}
    $\{ \ket{\phi_i} \}_{i=0}^m$ and $\{ \ket{\psi_j} \}_{j=0}^{n}$ {are} orthonormal bases on the input and output spaces of {$P$, respectively. Then we have}
    \begin{equation}
        \sum_n \vert \bra{\psi_{n+j}} {P} \ket{\phi_{n}} \vert \leq {\| P \|_{{\text{$k$-{\rm KF}}}}\quad\forall\,j. }
    \end{equation}
\end{lemma}
\begin{proof}
    {Firstly, we write $\bra{\psi_{r}} P \ket{\phi_{s}} = |\bra{\psi_{r}} P \ket{\phi_{s}}|e^{i \theta_{r,s}}$ $\forall\,r,s$. Now, consider the operator}
    $
        {X \coloneqq} \sum_{n} e^{i \theta_{n+j,n}} \ketbratwo{\psi_{n+j}}{\phi_n}.
    $
    It can then we seen that
    \begin{equation}
        {{\rm tr}\big( X^\dagger P \big)} = \sum_n \vert \bra{\psi_{n+j}} {P} \ket{\phi_{n}} \vert.
    \end{equation}
    It is noted that {\hbox{${\rm tr} \big( X^\dagger A \big) = \vert {\rm tr} \big( X^\dagger A \big) \vert$},} as it is non-negative.
    Now, the $k$th Ky-fan norm \cite{Fan1951} can be written as 
    \begin{equation}
        {\| P \|}_{{\text{$k$-{\rm KF}}}} = {\sup_{Z} \big\{ \vert {\rm tr} \big( Z^\dagger P \big) \vert : \| Z \|^*_{\scriptscriptstyle(k)} \leq 1 \big\},}
    \end{equation}
    where {$\| Z \|^*_{\scriptscriptstyle(k)}$} is the dual of the $k$th {Ky-Fan norm} \cite{Horn_Johnson_1985} and $Z$ is an operator on the Hilbert space. By definition, the dual of the $k$th {Ky-Fan norm} is
    \begin{equation}
        {\| Q \|}_{\scriptscriptstyle(k)}^* {\coloneqq} \sup_{Y} \big\{ \vert {\rm tr} {\big( Y^\dagger Q \big)} \vert : {\| Y \|}_{{\text{$k$-{\rm KF}}}} \leq 1 \big\}, \label{eq: dual ky fan norm}
    \end{equation}
    where $Y$ is an operator on the Hilbert space. To write the $k$th Ky-Fan norm in this way, we have used the fact that taking the dual of {a norm's dual} returns the norm \cite{Horn_Johnson_1985}. 

    Using the above definitions, it can be stated that 
     \begin{equation}
          \sum_n \vert \bra{\psi_{n+j}} {P} \ket{\phi_{n}} \vert = \vert  {\rm tr} {\big( X^\dagger P \big)} \vert = \leq {\| P \|_{\text{$k$-{\rm KF}}}}\quad{{\rm if}\quad \| X \|^*_{\scriptscriptstyle(k)} \leq 1.} \label{eq:bound_shown}
     \end{equation}
     To see that $\| X \|^*_{\scriptscriptstyle(k)} \leq 1$, we explicitly evaluate it. Firstly, one can first find the eigenvalues of $\sqrt{X^\dagger X}$, as these are the singular values of $X$. It can easily be seen that $X$ has $k$ singular values that are one, with the remaining being zero. 
     
     Now, let \hbox{${\| X \|}_{\scriptscriptstyle(k)}^* = \vert {\rm tr} \big[ (\CY{Y^*})^\dagger X \big] \vert$}, such that $Y^*$ is the operator satisfying Eq.~\eqref{eq: dual ky fan norm}. The {von Neumann’s} trace inequality \cite{von1937some} {can then} be employed to state that 
     \begin{equation}
         \begin{split}
             \| X \|_{\scriptscriptstyle(k)}^* = ~ \vert {\rm tr} \big[ (Y^*)^\dagger X \big] \vert &\leq \sum_{i=0}^{\min\{n,m\}} \mu_i(X) \mu_i(Y^*) \\
             &=  \sum_{i=0}^{k-1} \mu_i(Y^*) \eqqcolon \| Y^* \|_{k-{\rm KF}}\leq 1,
         \end{split}
     \end{equation}
     where, {for an operator $C$, the set} $\{\mu_i(C)\}_{i=0}^{\min\{n,m\}}$ is the singular values of $C$ ordered in decreasing order, such that $\mu_i(C) \geq \mu_i(C)~\forall~i$. To go from the first to second line, we have used the fact that {all non-zero singular values} of $X$ are one. We have then used the definintion of the $k$th Ky-fan norm: it is the sum over the $k$ largest singular values. The final {inequality} comes from the definition of the dual norm of the $k$th Ky-Fan norm, where $Y$ was restricted to have $ \| Y \|_{k-{\rm KF}}\leq 1 $ by definition. 
     Hence, $ \vert \vert X \vert \vert^*_{\scriptscriptstyle(k)} \leq 1$, {implying Eq.~\eqref{eq:bound_shown} and} completing the proof. 
\end{proof}

We now restate Result~\ref{result:bound_1} before {we begin} the proof. 
\\
\\
\textbf{Result~\ref{result:bound_1}}~\textit{{Let $d^{(j)}\coloneqq{\rm dim}(\mathcal{V}^{\modesuper{j}}_{\rm in})$.
For any finite-dimensional state $\rho$, we have
    \begin{equation}
           \Delta M_{A}^{\scriptscriptstyle(j)} \leq \| (\rho_A \otimes \rho_B)^{\scriptscriptstyle(j)} \|_{\text{{${d}^{(j)}$}-{\rm KF}}} - \onenorm{\rho^{\scriptscriptstyle(j)}},
    \end{equation}
where $\| P \|_{\text{$\alpha$-{\rm KF}}}$ is the $\alpha$-Ky-Fan norm of the operator $P$, which is the sum of the $\alpha$ largest singular values of $P$.}
}
\begin{proof} 
    The $j$th bipartite mode {subspace,} $\mathcal{B}^{\scriptscriptstyle(j)}_{AB}$, has {$d^{(j)}\coloneqq{\rm dim}(\mathcal{V}^{\modesuper{j}}_{\rm in})$} basis operators {whose} coefficients will contribute to the local mode upon performing the partial trace. The rough goal of coherence concentration is to maximise the value of these coefficients under conjugation by {an allowed unitary.} It is noted {that} each basis operator in {$\mathcal{V}_{\rm in}^{(j)}$} is in a distinct row and column along a diagonal of the $j$th bipartite mode operator.
    
    Now, $M^{\scriptscriptstyle(j)}(\sigma_A){\coloneqq \onenorm{\sigma_A^{\scriptscriptstyle(j)}}}$ is the sum of the absolute values of the matrix elements of $\sigma_A$ in the $j$th mode. Each matrix element of $\sigma_A^{\scriptscriptstyle(j)}$ is then a sum over some subset of the coefficients of the basis states in {$\mathcal{V}_{\rm in}^{(j)}$} for the $j$th bipartite mode $\sigma_{AB}^{\scriptscriptstyle(j)}$. Hence, using the triangle inequality, it can be seen that 
    \begin{equation}
        \begin{split}
            M^{\scriptscriptstyle(j)}(\sigma_A) &\leq \sum_{n,m} {\vert \bra{n+j, m} \sigma_{AB}^{\scriptscriptstyle(j)} \ket{n,m} \vert,} \\
        \end{split}
    \end{equation}
    where 
    {$\bra{n+j, m} \sigma_{AB}^{\scriptscriptstyle(j)} \ket{n,m}$} are the {$d^{(j)}$} coefficients of {$\sigma_{AB}^{\scriptscriptstyle(j)}$ with respect to} the basis {operators} in {$\mathcal{V}_{\rm in}^{(j)}$.} This {upper bound} is then a sum over the absolute values of the matrix elements of a diagonal of the $j$th bipartite mode operator, $\sigma_{AB}^{\scriptscriptstyle(j)}$, in the eigenbasis of $L_{AB}$. {Lemma~\ref{lemma_sum_coefficents}} can then be employed to {imply that} 
    \begin{equation}
        \begin{split}
            \sum_{n,m} {\vert \bra{n+j, m} \sigma_{AB}^{\scriptscriptstyle(j)} \ket{n,m} \vert \leq \|   \sigma_{AB}^{\scriptscriptstyle(j)} \|_{\text{$d^{(j)}$-{\rm KF}}} = \| (\rho_A \otimes \rho_B)^{\scriptscriptstyle(j)} \|_{\text{$d^{(j)}$-{\rm KF}}},}
        \end{split}
    \end{equation}
    as the {Ky-Fan} norm is unitarily invariant. Taking away $M^{\scriptscriptstyle(j)}(\rho_A)$ from both sides completes the proof. 
\end{proof}

\subsection{Proof of Result~\ref{result:bound_2} \label{appendix:result_2_proof}}

{Here, we present} a proof of Result~\ref{result:bound_2}, which is first stated {for clarity.} 
\\
\\
\textbf{{Result~\ref{result:bound_2}}}~\textit{
{Let $d^{(j,c)} \coloneqq {\rm dim}(\mathcal{V}_{\rm in}^{\modesuper{j,c}})$.
For any $d$-dimensional state state $\rho$, we have}
\begin{equation}
     \Delta M_{A}^{\scriptscriptstyle(j)} \leq \sum_{c=0}^{2d-2-j} \| \Pi_{c+j} (\rho \otimes \rho)^{\scriptscriptstyle(j)} \Pi_{c} \|_{\text{{$d^{(j,c)}$}-{\rm KF}}} - \onenorm{\rho^{\modesuper{j}}},
\end{equation}
{where $\Pi_{c}$ is the projection onto the eigenspace of $L_{AB}$ with eigenvalue $\lambda_c$.}}
\begin{proof}
    To prove this, one firstly needs to show that each LR-D subspace transforms independently. {From} here, the proof of Result~\ref{result:bound_1} can then be applied to each individual space.  
    To see the independence of each LR-D subspace, recall that {an allowed unitary will apply independent unitaries within} each degenerate subspace of the operator defining our chosen basis. In the concentration scenario, {a global allowed unitarity} $V_{AB}$ therefore {acts as an} independent unitary within each degenerate space of {the operator $L_{AB} = \sum_c\lambda_c\Pi_c$ with spectrum $\{\lambda_c\}_c$. More precisely, for each $c$, there exist unitary operators $V_c$ such that} 
    {\begin{align}
        V_{AB} = \sum_{c=0}^{2d-2} V_{c} \Pi_c,
        \end{align}
        where the projector $\Pi_c$ is given by
        \begin{align}
        \Pi_c = {\sum_{n+m = k+j = c}} \ketbratwo{nm}{kl}.
    \end{align}} 
    A $j$th mode operator, as seen in the main text, is some operator supported in {the} $j$th mode subspace, 
\begin{equation}
    \sigma^{\modesuper{j}}_{AB} \in \mathcal{B}_{AB}^{\modesuper{j}} = {\rm span}\{ \ketbratwo{n+k, m+j-k}{n,m} : \forall~n,m,k \},
\end{equation}
    such that 
    \begin{equation}
        \begin{split}
            \sigma_{AB}^{\modesuper{j}} =  \sum_{k} \sum_{n,m} q_{n+k,m+j-k,n,m} \ketbratwo{n+k, m+j-k}{n,m} = {\sum_{n, \alpha, \beta} \Tilde{q}_{n+j, n, \alpha, \beta} \ketbratwo{n+j, \alpha}{n, \beta},}
        \end{split}
    \end{equation}
    where {we have performed} a relabelling to reflect the degeneracies in $L_{AB}$, using the dummy indices $\alpha, \beta$ to label the degeneracies. From here, it can be seen that 
    \begin{equation}
        \begin{split}
            V_{AB}{\sigma_{AB}^{\modesuper{j}}}V_{AB}^\dagger &= \bigg( \sum_{c=0}^{2d-2} V_{c}  \Pi_c \bigg) \sigma_{AB}^{\modesuper{j}} \bigg(\sum_{e=0}^{2d-2}  \Pi_e V_{e}^\dagger \bigg)= \sum_{c,e=0}^{2d-2} \sum_{n, \alpha, \beta} \Tilde{q}_{n+j, n, \alpha, \beta} ~ V_c \Pi_c  \ketbratwo{n+j, \alpha}{n, \beta} \Pi_e V_e^\dagger \\
            &= \sum_{n, \alpha, \beta}  \Tilde{q}_{n+j, n, \alpha, \beta} V_{n+j} \ketbratwo{n+j, \alpha}{n, \beta} V_{n}^\dagger = \sum_n V_{n+j} \Pi_{n+j}{\sigma^{\modesuper{j}}_{AB}}\Pi_{n} V_n.
        \end{split}
    \end{equation}
    where in {the first line} we have used the fact that \hbox{$\Pi_c \ket{n, \alpha} = \delta_{c,n} \ket{n, \alpha}~\forall~\alpha$}, and in {the second line} we have used
    \begin{equation}
        \begin{split}
            \Pi_{c+j} {\sigma^{\modesuper{j}}_{AB}}\Pi_{c} = \sum_{\alpha, \beta} \Tilde{q}_{c+j, c, \alpha, \beta} \ketbratwo{c+j, \alpha}{c, \beta}. 
        \end{split}
    \end{equation}
    As the unitary $V_c$ maps states within the $c$th degenerate subspace back to $c$th degenerate subspace, it can be concluded that each $\Pi_{c+j} {\sigma^{\modesuper{j}}_{AB}}\Pi_{c}$, {that is, each LR-D subspace,} transforms independently under {allowed unitaries} $V_{AB}$. We note that {the} $c$th degenerate subspace of $L_{AB}$ has dimension 
\begin{equation}
    {d(c,L_{AB}) \coloneqq} \begin{cases}
        c+1 ~ ~ &0\leq c < d-1, \\
        2d-1-c ~ ~ &d-1\leq c \leq 2d-2.
    \end{cases}
\end{equation} 
    {Hence, the $(c+j,c)$-LR-D operator has its input and output dimensions as $d(c+j,L_{AB})$ and $d(c,L_{AB})$, respectively.}
    {Also, there are $2d-2-j$ LR-D operators} within the $j$th global mode.

    Returning to the original labelling, one can define the $(c+j, c)$-LR-D subspace of the $j$th mode space as
    \begin{equation}
        \mathcal{B}_{AB}^{\modesuper{j,c}} {\coloneqq {\rm span}}\{ \ketbratwo{n+k, m+j-k}{n,m} : n+m=c \} \subset \mathcal{B}_{AB}^{\modesuper{j}}.
    \end{equation}
    As before, this can be split into {two subspaces,} one where the components make a contribution to the local mode and where they do not
    \begin{equation}
        \begin{split}
            \mathcal{V}^{\modesuper{j,c}}_{\mathrm{in}} &{\coloneqq {\rm span}}\{ \ketbratwo{n+j, m}{n,m} : n+m=c \} {\subseteq  \mathcal{B}_{AB}^{\modesuper{j,c}};}  \\
            \mathcal{V}^{\modesuper{j,c}}_{\mathrm{out}} &{\coloneqq {\rm span}}\{ \ketbratwo{n+k, m+j-k}{n,m} : n+m=c, k \neq j \} {\subseteq  \mathcal{B}_{AB}^{\modesuper{j,c}},}
        \end{split}
    \end{equation}
    where the aim of mode concentration is still {to maximise} the coefficients {with respect to the basis operators} in $\mathcal{V}^{\modesuper{j,c}}_{\mathrm{in}}$. Given {that each of the LR-D subspaces evolves} unitarily, the proof of Result~\ref{result:bound_1} can then be applied to each LD-R subspace individually. 
\end{proof}

\begin{figure}[h!]
    \centering
    \begin{minipage}{0.48\textwidth}
        \centering
        \includegraphics[width=\linewidth]{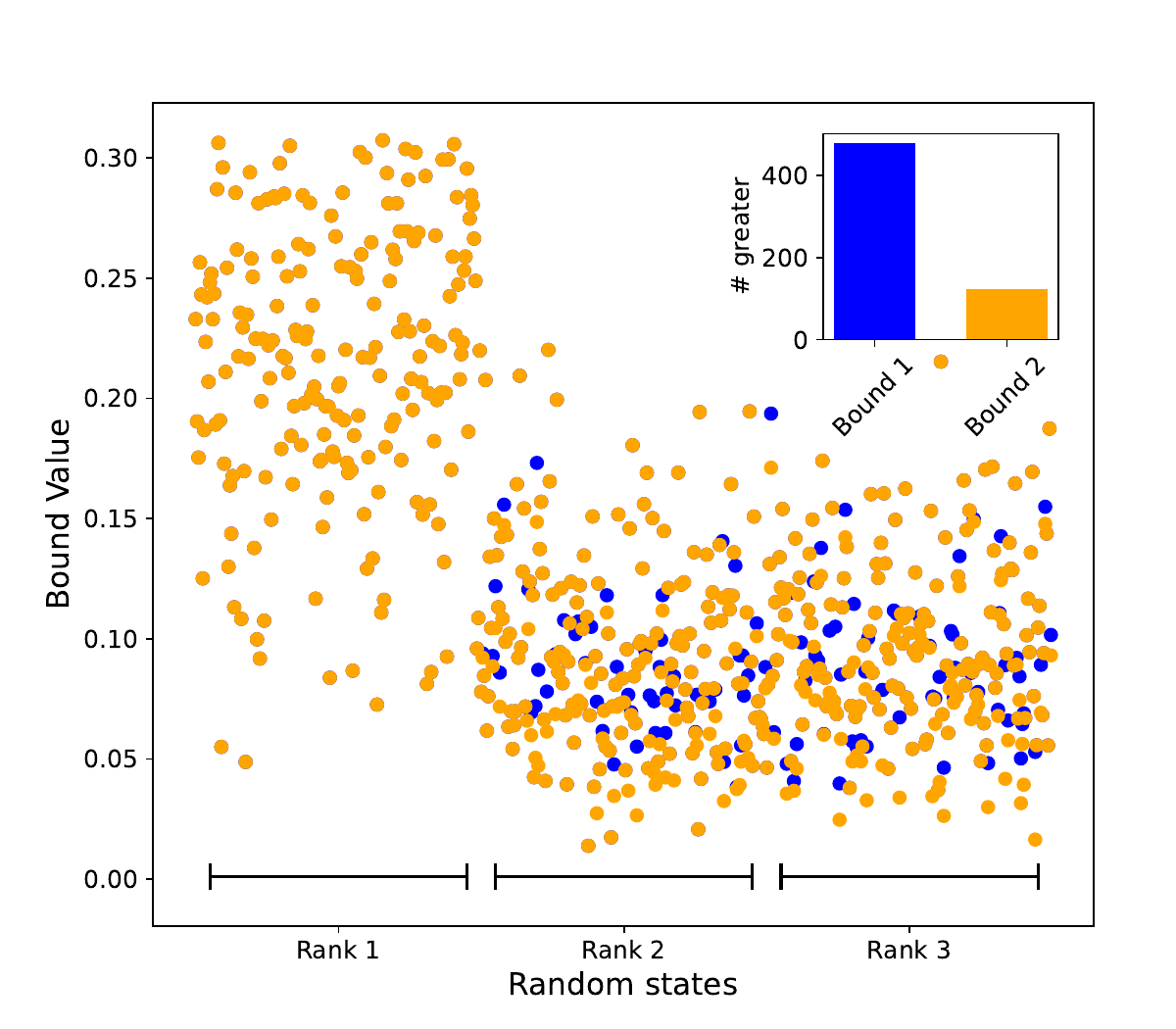}
        \caption{The value of Result~\ref{result:bound_1} (bound 1) and Result~\ref{result:bound_2} (bound 2) plotted for a sample of states in $d=3$ for all possible ranks. The insert shows the number of sampled states for which a given bound is greater than the other.}
        \label{fig: bound_compare_dim_3}
    \end{minipage}
    \hfill
    \begin{minipage}{0.48\textwidth}
        \centering
        \includegraphics[width=\linewidth]{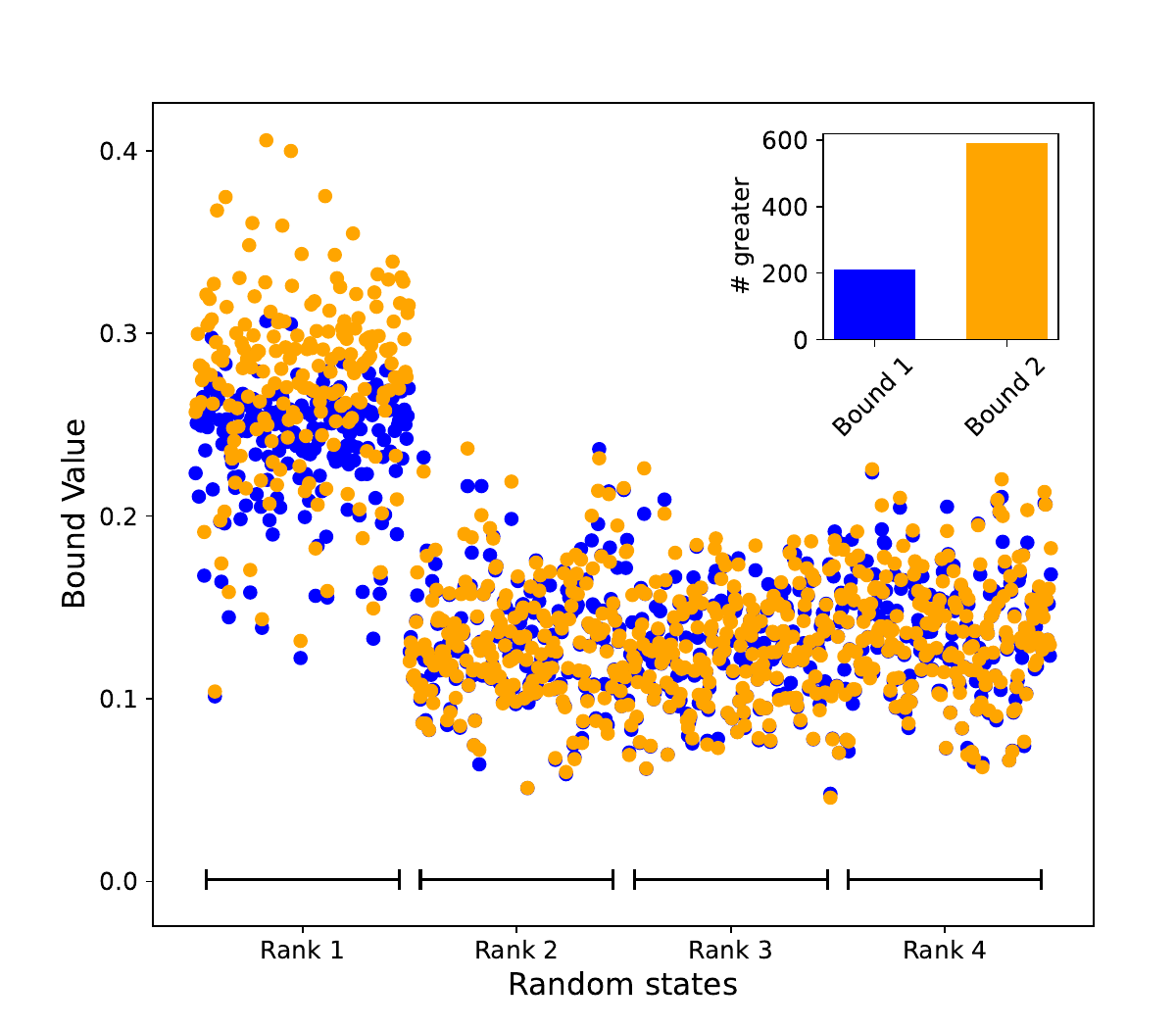}
        \caption{The value of Result~\ref{result:bound_1} (bound 1) and Result~\ref{result:bound_2} (bound 2) plotted for a sample of states in $d=4$ for all possible ranks. The insert shows the number of sampled states for which a given bound is greater than the other.}
        \label{fig: bound_compare_dim_4}
    \end{minipage}
\end{figure}
\section{Bound Comparison} \label{appendix: bound comparision}
Here, the {bounds on $\Delta M_{A}^{\scriptscriptstyle(j)}$ given in Result~\ref{result:bound_1} and Result~\ref{result:bound_2}} are compared for {$d =3$ and $d=4$} for a sample of random states {varying in rank (see Figs.~\ref{fig: bound_compare_dim_3} and~\ref{fig: bound_compare_dim_4}).} It is firstly noted that for {$d=3$, both} bounds are equal for all tested {pure (i.e., rank-one)} states. {For states with higher ranks in $d=3$, and states in $d=4$ (including rank-one),} we see that there is not a clear comparison of the two upper bounds. This can be confirmed in the insert, which shows the number of sampled states for which one bound is greater {than the other. In general, it can be seen that Result~\ref{result:bound_2} performs better for $d=3$, and  Result~\ref{result:bound_1} performs better for $d=4$.} We conjecture that these bounds will not coincide for all $d$, and that in general one should just take whichever of the two bounds is smallest to be the upper bound on the concentrateable coherence.

\section{A Sufficient Condition for The existence of Correlations} \label{appendix: A Sufficient Condition for The existence of Correlations}

Here, we show {that} if a bipartite state {in $AB$ (with equal local dimension $d$)} has support only on the $0$th mode and modes $k\geq d$, then {there must} exist correlation between {$A$ and $B$. Formally, we prove that:}
\begin{lemma}\label{lemma: A Sufficient Condition for The existence of Correlations}
{Let $\rho_{AB}$ be a bipartite state in $AB$ with equal local dimension $d<\infty$.
Then $\rho_{AB}$ is not a product state, and hence has correlation, if}
    \begin{equation}
        {\rm modes}(\rho_{AB}) \neq \{0\} ~ ~  but ~ ~ [1, d-1] \cap {\rm modes}(\rho_{AB}) = \emptyset.
    \end{equation}
    {Namely, $\rho_{AB}$ satisfies} 
    \begin{equation}
    \rho_{AB} = \rho_{AB}^{(0)} + \sum_{i \geq d}\rho_{AB}^{(i)}\quad {\&\quad\sum_{i \geq d}\rho_{AB}^{(i)}\neq0.} 
    \end{equation}
\end{lemma}
\begin{proof}
    {Suppose the opposite, that $\rho_{AB}=\rho_A\otimes\rho_B$ and} 
    \begin{equation}
        \rho_A \otimes \rho_B = \rho_{AB}^{(0)} + \sum_{i \geq d}\rho_{AB}^{(i)}\quad{\&\quad\sum_{i \geq d}\rho_{AB}^{(i)}\neq0.} \label{eq: proof by contradiction}
    \end{equation}
    Due to Lemma~\ref{global_Modes_to_local_modes}, it can then be concluded that \hbox{$\rho_A = {\rm tr}_B{\big(\rho_{AB}^{\modesuper{0}}\big)} = \rho_A^{\modesuper{0}}$} and \hbox{$\rho_B = {\rm tr}_A{\big(\rho_{AB}^{\modesuper{0}}\big)} = \rho_B^{\modesuper{0}}$}, meaning that both $\rho_A$ and $\rho_B$ would need to be states supported only on the $0$th mode. However, using the fact that the bipartite modes of a tensor product of states {are given by~\cite{PhysRevA.90.062110}:}
    \begin{equation}
        (\rho_A \otimes \rho_B)^{\modesuper{j}} = \sum_{k} \rho_A^{\modesuper{k}} \otimes \rho_B^{\modesuper{j-k}}\quad{\forall\,j,}
    \end{equation}
    {it can be concluded that $\rho_{AB} = \rho_A^{\modesuper{0}} \otimes \rho_B^{\modesuper{0}} = \rho_{AB}^{\modesuper{0}}$,
    meaning that $\sum_{i \geq d}\rho_{AB}^{(i)}=0$, in contradiction with Eq.~\eqref{eq: proof by contradiction}.}
\end{proof}

\section{Qutrit Example}\label{Section: qutrit example sup mat}

Here we will run through an example of mode concentration using two qutrits. We will use graphical deceptions where appropriate to make our explanations as clear as possible.

Let {$\rho$ be a qutrit state,} and $L = \sum_{n=0}^{2} n \ketbra{n}$ be the Hermitian operator defining the eigenspaces between which we aim to increase coherence. The $j$th mode of $\rho$ is define as $\rho^{\modesuper{j}} = \sum_{n} p_{n+j,n} \ketbratwo{n+j}{n}$, with all the modes of $\rho$ be visualisable as
\NiceMatrixOptions{cell-space-limits = 1pt}
\begin{equation}
    \rho = 
\begin{bNiceArray}{>{\strut}ccc}[margin,extra-margin = 3pt, columns-width = 8pt]
p_{00} & p_{01} & p_{02}  \\
p_{10} & p_{11} & p_{12} \\
p_{20} & p_{21} & p_{22} \\
\CodeAfter
    \begin{tikzpicture}
    \tikzset{main diagonal highlight/.style={
        rounded corners=5pt, 
        thick,               
        fill=#1!100,          
        inner sep=-1pt,       
        rotate fit=-30,      
        fill opacity=0.4,    
    }}

    \node [main diagonal highlight=red, fit = (1-1) (3-3), rotate fit=-15, inner sep=-2.5pt] {};
    \node [main diagonal highlight=blue, fit = (1-2) (2-3), inner sep=-1.2pt] {};
    \node [main diagonal highlight=green, fit = (1-3), inner sep=0pt] {};
    \node [main diagonal highlight=yellow, fit = (2-1) (3-2), inner sep=-1.5pt, rotate fit=-15] {};
    \node [main diagonal highlight=orange, fit = (3-1) (3-1), inner sep=-1pt] {};

    \end{tikzpicture}
\end{bNiceArray} \qquad 
    \text{where } 
    \qquad
    \begin{array}{c|c}
        \text{Mode} & \text{Colour} \\ 
        \hline \\ [-1em]
        \rho^{\modesuper{0}} & \drawCustomCircle{0}{0}{0.1cm}{red}{red}{0.4}{0pt}{} \\
        \rho^{\modesuper{1}} & \drawCustomCircle{0}{0}{0.1cm}{yellow}{yellow}{0.4}{0pt}{} \\
        \rho^{\modesuper{2}} & \drawCustomCircle{0}{0}{0.1cm}{orange}{orange}{0.4}{0pt}{} \\
        \rho^{\modesuper{-1}} & \drawCustomCircle{0}{0}{0.1cm}{blue}{blue}{0.4}{0pt}{} \\
        \rho^{\modesuper{-2}} & \drawCustomCircle{0}{0}{0.1cm}{green}{green}{0.4}{0pt}{} \\
    \end{array}
\end{equation}
such that, {for example,} 
\begin{equation}
    \rho^{\modesuper{0}} = \begin{bNiceArray}{>{\strut}ccc}[margin,extra-margin = 3pt, columns-width = 8pt]
    p_{00} & 0 & 0  \\
    0 & p_{11} & 0  \\
    0 & 0 & p_{22}  \\
    \end{bNiceArray} 
    \qquad
    \textnormal{and} \qquad
    \rho^{\modesuper{2}} = \begin{bNiceArray}{>{\strut}ccc}[margin,extra-margin = 3pt, columns-width = 8pt]
    0 & 0 & 0  \\
    0 & 0 & 0  \\
    p_{20} & 0 & 0  \\
    \end{bNiceArray}  .
\end{equation}

We will here aim to concentrate coherence onto the $2$nd mode, detailing the framework presented in the main text.
{To do so, one needs to apply an allowed unitary} to the $2$nd global mode and then {perform} the partial trace. The $2$nd global mode is given by 
\begin{equation}
    (\rho_A \otimes \rho_B)^{\modesuper{2}} = \rho_A^{\modesuper{0}} \otimes \rho_B^{\modesuper{2}} + \rho_A^{\modesuper{1}} \otimes \rho_B^{\modesuper{1}} + \rho_A^{\modesuper{2}} \otimes \rho_B^{\modesuper{0}},
\end{equation}
which, with respect to the eigenbasis of $L_{AB}$, is given by 
\begin{equation}
    (\rho_A \otimes \rho_B)^{\modesuper{2}} = \begin{bNiceArray}{>{\strut}ccccccccc}[margin,extra-margin = 3pt, columns-width = 25pt]
    0 & 0 & 0 & 0 & 0 & 0 & 0 & 0 & 0 \\
    0 & 0 & 0 & 0 & 0 & 0 & 0 & 0 & 0 \\
    \drawCustomCircletwo{0}{0}{0.25cm}{orange}{orange}{0.4}{0pt}{$p_{00}p_{20}$} & 0 & 0 & 0 & 0 & 0 & 0 & 0 & 0 \\
    0 & 0 & 0 & 0 & 0 & 0 & 0 & 0 & 0 \\
    \drawCustomCircletwo{0}{0}{0.25cm}{orange}{orange}{0.4}{0pt}{$p_{10}p_{10}$} & 0 & 0 & 0 & 0 & 0 & 0 & 0 & 0 \\
    0 &  \drawCustomCircletwo{0}{0}{0.25cm}{orange}{orange}{0.4}{0pt}{$p_{10}p_{21}$} & 0 &  \drawCustomCircletwo{0}{0}{0.25cm}{orange}{orange}{0.4}{0pt}{$p_{11}p_{20}$} & 0 & 0 & 0 & 0 & 0 \\
     \drawCustomCircletwo{0}{0}{0.25cm}{pink}{pink}{0.4}{0pt}{$p_{20}p_{00}$} & 0 & 0 & 0 & 0 & 0 & 0 & 0 & 0 \\
    0 &  \drawCustomCircletwo{0}{0}{0.25cm}{pink}{pink}{0.4}{0pt}{$p_{20}p_{11}$} & 0 &  \drawCustomCircletwo{0}{0}{0.25cm}{orange}{orange}{0.4}{0pt}{$p_{21}p_{10}$} & 0 & 0 & 0 & 0 & 0 \\
    0 & 0 &  \drawCustomCircletwo{0}{0}{0.25cm}{pink}{pink}{0.4}{0pt}{$p_{20}p_{22}$} & 0 &  \drawCustomCircletwo{0}{0}{0.25cm}{orange}{orange}{0.4}{0pt}{$p_{21}p_{21}$} & 0 &  \drawCustomCircletwo{0}{0}{0.25cm}{orange}{orange}{0.4}{0pt}{$p_{22}p_{20}$} & 0 & 0 \\
    \end{bNiceArray} 
    \qquad
        \text{} 
    \qquad
    \begin{array}{c|c}
        \text{Space} & \text{Colour} \\ 
        \hline \\ [-1em]
        {\mathcal{V}_{\rm in}^{(2)}} & \drawCustomCircle{0}{0}{0.1cm}{pink}{pink}{0.4}{0pt}{} \\
        {\mathcal{V}_{\rm out}^{(2)}} & \drawCustomCircle{0}{0}{0.1cm}{orange}{orange}{0.4}{0pt}{} \\
    \end{array}
\end{equation}
where we have highlighted the coefficients associated {with basis operators in {$\mathcal{V}_{\rm in}^{(2)}$,} i.e., the basis operators whose} coefficients contribute to the local mode. In this case, the $2$nd local mode has a single component, given by the sum of the coefficients associated {with the basis operators in {$\mathcal{V}_{\rm in}^{(2)}$}} (the pink dots). Hence, {$M^{\modesuper{2}}(\sigma_A)$} is the absolute value of the sum of these components. Optimally concentrating coherence on the $2$nd mode therefore means maximising the absolute value of the sum of the coefficients {associated with the basis operators in {$\mathcal{V}_{\rm in}^{(2)}$}.}  An upper-bound on {$M^{\modesuper{2}}(\sigma_A)$} is then given by the sum of the maximised absolute values of each coefficient, due to the triangle inequality.

    {
    \subsection{Qutrit Example for Result~\ref{result:bound_1}}
    }
{We now} detail the logical steps that led to Result~\ref{result:bound_1} in the main text. Firstly, we relax the problem by assuming that any {pre- and post-processing} unitary can be applied to $(\rho_A \otimes \rho_B)^{\modesuper{2}}$, not just a conjugation by {an allowed unitary.} From here, Lemma~\ref{lemma_sum_coefficents} can be employed to show that the sum over the absolute value of the matrix coefficients that give the $2$nd local mode is upper-bounded by the $3$rd Ky-Fan norm of this operator. This is because we are summing over a diagonal of length ${\rm dim}({\mathcal{V}_{\rm in}^{(2)}})=3$. {As} {$M^{\modesuper{2}}(\sigma_A)$} is upper-bounded by the sum of the absolute values of these matrix coefficients, 
{which is further} upper-bounded by the $3$rd Ky-Fan norm of this operator, {one can therefore conclude that \hbox{$M^{\modesuper{2}}(\sigma_A) \leq \| (\rho_A \otimes \rho_B)^{\modesuper{2}} \|_{\text{$3$-{\rm KF}}}$}. }

    {
    \subsection{Qutrit Example for Result~\ref{result:bound_2}}
    }

We now consider in detail splitting $(\rho_A \otimes \rho_B)^{\modesuper{2}}$ up into the individual LR-D subspaces that evolve independently under {allowed unitaries,} and hence go {through} the logical steps that led to Result~\ref{result:bound_2} in the main text. 
Firstly, as stated in {Section~\ref{appendix:result_2_proof},} any {allowed} unitary can be decomposed as {\hbox{$V_{AB} = \sum_{c} V_c \Pi_c$}, where} each $V_k$ is {an independent unitary, if considering coherence between the eigenspaces of the operator $L_{AB}=\sum_c \lambda_c \Pi_c$ with spectrum $\{\lambda_c\}_c$.} Here, we consider the Hermitian operator $L_{AB} = L_A \otimes \mathbb{I}_B + \mathbb{I}_A \otimes L_B$, such that with respect to the eigenbasis of $L_{AB}$ any {allowed} unitary can be decomposed as   

\begin{equation}
    V_{AB} = \begin{bNiceArray}{>{\strut}ccccccccc}[margin,extra-margin = 3pt, columns-width = 7pt]
    \drawCustomCircletwo{0}{0}{0.25cm}{yellow}{yellow}{0.8}{0pt}{} & 0 & 0 & 0 & 0 & 0 & 0 & 0 & 0 \\
    0 &\drawCustomCircletwo{0}{0}{0.25cm}{magenta}{magenta}{0.8}{0pt}{} & 0 & \drawCustomCircletwo{0}{0}{0.25cm}{magenta}{magenta}{0.8}{0pt}{} & 0 & 0 & 0 & 0 & 0 \\
    0 & 0 & \drawCustomCircletwo{0}{0}{0.25cm}{RedOrange}{RedOrange}{0.8}{0pt}{} & 0 & \drawCustomCircletwo{0}{0}{0.25cm}{RedOrange}{RedOrange}{0.8}{0pt}{} & 0 & \drawCustomCircletwo{0}{0}{0.25cm}{RedOrange}{RedOrange}{0.8}{0pt}{} & 0 & 0 \\
    0 & \drawCustomCircletwo{0}{0}{0.25cm}{magenta}{magenta}{0.8}{0pt}{} & 0 & \drawCustomCircletwo{0}{0}{0.25cm}{magenta}{magenta}{0.8}{0pt}{} & 0 & 0 & 0 & 0 & 0 \\
    0 & 0 & \drawCustomCircletwo{0}{0}{0.25cm}{RedOrange}{RedOrange}{0.8}{0pt}{} & 0 &\drawCustomCircletwo{0}{0}{0.25cm}{RedOrange}{RedOrange}{0.8}{0pt}{} & 0 & \drawCustomCircletwo{0}{0}{0.25cm}{RedOrange}{RedOrange}{0.8}{0pt}{} & 0 & 0 \\
    0 &  0 & 0 &  0 & 0 & \drawCustomCircletwo{0}{0}{0.25cm}{cyan}{cyan}{1}{0pt}{} & 0 & \drawCustomCircletwo{0}{0}{0.25cm}{cyan}{cyan}{0.8}{0pt}{} & 0 \\
    0 & 0 & \drawCustomCircletwo{0}{0}{0.25cm}{RedOrange}{RedOrange}{0.8}{0pt}{} & 0 & \drawCustomCircletwo{0}{0}{0.25cm}{RedOrange}{RedOrange}{0.8}{0pt}{} & 0 & \drawCustomCircletwo{0}{0}{0.25cm}{RedOrange}{RedOrange}{0.8}{0pt}{} & 0 & 0 \\
    0 & 0 & 0 &  0 & 0 & \drawCustomCircletwo{0}{0}{0.25cm}{cyan}{cyan}{0.8}{0pt}{} & 0 & \drawCustomCircletwo{0}{0}{0.25cm}{cyan}{cyan}{0.8}{0pt}{} & 0 \\
    0 & 0 & 0 & 0 &  0 & 0 &  0 & 0 & \drawCustomCircletwo{0}{0}{0.25cm}{violet}{violet}{0.8}{0pt}{} \\
    \end{bNiceArray} 
    \qquad
        \text{} 
    \qquad
    \begin{array}{c|c}
        \text{Unitary} & \text{Colour} \\ 
        \hline \\ [-1em]
        V_0 & \drawCustomCircle{0}{0}{0.1cm}{yellow}{yellow}{0.8}{0pt}{} \\
        V_1 & \drawCustomCircle{0}{0}{0.1cm}{magenta}{magenta}{0.8}{0pt}{} \\
        V_2 & \drawCustomCircle{0}{0}{0.1cm}{RedOrange}{RedOrange}{0.8}{0pt}{} \\
        V_3 & \drawCustomCircle{0}{0}{0.1cm}{cyan}{cyan}{0.8}{0pt}{} \\
        V_4 & \drawCustomCircle{0}{0}{0.1cm}{violet}{violet}{0.8}{0pt}{} \\

    \end{array}
\end{equation}
where each matrix {element associated with a given} independent unitary is represented by a colour. This unitary is then applied to \hbox{$(\rho_A \otimes \rho_B)^{\modesuper{2}}$}. 
We now represent \hbox{$(\rho_A \otimes \rho_B)^{\modesuper{2}}$} while highlighting all the different LR-D subspaces with the mode,
\begin{equation}
    (\rho_A \otimes \rho_B)^{\modesuper{2}} = \begin{bNiceArray}{>{\strut}ccccccccc}[margin,extra-margin = 3pt, columns-width = 25pt]
    0 & 0 & 0 & 0 & 0 & 0 & 0 & 0 & 0 \\
    0 & 0 & 0 & 0 & 0 & 0 & 0 & 0 & 0 \\
    \drawCustomCircletwo{0}{0}{0.25cm}{Plum}{Plum}{0.4}{0pt}{$p_{00}p_{20}$} & 0 & 0 & 0 & 0 & 0 & 0 & 0 & 0 \\
    0 & 0 & 0 & 0 & 0 & 0 & 0 & 0 & 0 \\
    \drawCustomCircletwo{0}{0}{0.25cm}{Plum}{Plum}{0.4}{0pt}{$p_{10}p_{10}$} & 0 & 0 & 0 & 0 & 0 & 0 & 0 & 0 \\
    0 &  \drawCustomCircletwo{0}{0}{0.25cm}{teal}{teal}{0.4}{0pt}{$p_{10}p_{21}$} & 0 &  \drawCustomCircletwo{0}{0}{0.25cm}{teal}{teal}{0.4}{0pt}{$p_{11}p_{20}$} & 0 & 0 & 0 & 0 & 0 \\
     \drawCustomCircletwo{0}{0}{0.25cm}{Plum}{Plum}{0.4}{0pt}{$p_{20}p_{00}$} & 0 & 0 & 0 & 0 & 0 & 0 & 0 & 0 \\
    0 &  \drawCustomCircletwo{0}{0}{0.25cm}{teal}{teal}{0.4}{0pt}{$p_{20}p_{11}$} & 0 &  \drawCustomCircletwo{0}{0}{0.25cm}{teal}{teal}{0.4}{0pt}{$p_{21}p_{10}$} & 0 & 0 & 0 & 0 & 0 \\
    0 & 0 &  \drawCustomCircletwo{0}{0}{0.25cm}{Bittersweet}{Bittersweet}{0.4}{0pt}{$p_{20}p_{22}$} & 0 &  \drawCustomCircletwo{0}{0}{0.25cm}{Bittersweet}{Bittersweet}{0.4}{0pt}{$p_{21}p_{21}$} & 0 &  \drawCustomCircletwo{0}{0}{0.25cm}{Bittersweet}{Bittersweet}{0.4}{0pt}{$p_{22}p_{20}$} & 0 & 0 \\
    \end{bNiceArray} 
    \qquad
        \text{} 
    \qquad
    \begin{array}{c|c}
        \text{LR-D Space} & \text{Colour} \\ 
        \hline \\ [-1em]
        (2,0) & \drawCustomCircle{0}{0}{0.1cm}{Plum}{Plum}{0.4}{0pt}{} \\
        (3,1) & \drawCustomCircle{0}{0}{0.1cm}{teal}{teal}{0.4}{0pt}{} \\
        (4,2) & \drawCustomCircle{0}{0}{0.1cm}{Bittersweet}{Bittersweet}{0.4}{0pt}{} \\
    \end{array}
\end{equation}
where the $(c+j,c)$ LR-D subspace of the $j$th mode is defined as $\Pi_{c+j}(\rho_A \otimes \rho_B)^{\modesuper{j}}\Pi_{c}$. 
{As stated in the main text,} the benefit of this formulation is that each LR-D subspace evolves independently and unitarily. More specifically, each of the LR-D subspaces evolve as 
\begin{align}
    V_{2} {\Pi_{2}(\rho_A \otimes \rho_B)^{\modesuper{2}}\Pi_{0}} V_{0} 
    &= 
    \begin{bNiceArray}{>{\strut}ccc}[margin,extra-margin = 3pt, columns-width = 8pt]
        \drawCustomCircletwo{0}{0}{0.25cm}{RedOrange}{RedOrange}{0.8}{0pt}{} &  
        \drawCustomCircletwo{0}{0}{0.25cm}{RedOrange}{RedOrange}{0.8}{0pt}{} &  
        \drawCustomCircletwo{0}{0}{0.25cm}{RedOrange}{RedOrange}{0.8}{0pt}{}  \\
        \drawCustomCircletwo{0}{0}{0.25cm}{RedOrange}{RedOrange}{0.8}{0pt}{} &  
        \drawCustomCircletwo{0}{0}{0.25cm}{RedOrange}{RedOrange}{0.8}{0pt}{} &  
        \drawCustomCircletwo{0}{0}{0.25cm}{RedOrange}{RedOrange}{0.8}{0pt}{}  \\
        \drawCustomCircletwo{0}{0}{0.25cm}{RedOrange}{RedOrange}{0.8}{0pt}{} &  
        \drawCustomCircletwo{0}{0}{0.25cm}{RedOrange}{RedOrange}{0.8}{0pt}{} &  
        \drawCustomCircletwo{0}{0}{0.25cm}{RedOrange}{RedOrange}{0.8}{0pt}{}  \\
    \end{bNiceArray}
    \begin{bNiceArray}{>{\strut}ccc}[margin,extra-margin = 3pt, columns-width = 0pt]
        p_{00}p_{20} \\
        p_{10}p_{10} \\
        \drawCustomCircletwo{0}{0}{0.25cm}{pink}{pink}{0.8}{0pt}{$p_{20}p_{00}$} \\
    \end{bNiceArray}
    \begin{bNiceArray}{>{\strut}c}[margin,extra-margin = 3pt, columns-width = 0pt]
        \drawCustomCircletwo{0}{0}{0.25cm}{yellow}{yellow}{0.8}{0pt}{}
    \end{bNiceArray}
    \\
    V_{3} {\Pi_{2}(\rho_A \otimes \rho_B)^{\modesuper{2}}\Pi_{0}} V_{1} 
    &= 
    \begin{bNiceArray}{>{\strut}cc}[margin,extra-margin = 0pt, columns-width = -1pt]
    \drawCustomCircletwo{0}{0}{0.25cm}{cyan}{cyan}{0.8}{0pt}{} &  
    \drawCustomCircletwo{0}{0}{0.25cm}{cyan}{cyan}{0.8}{0pt}{} \\
    \drawCustomCircletwo{0}{0}{0.25cm}{cyan}{cyan}{0.8}{0pt}{} &  
    \drawCustomCircletwo{0}{0}{0.25cm}{cyan}{cyan}{0.8}{0pt}{} \\
    \end{bNiceArray}
    \begin{bNiceArray}{>{\strut}ccc}[margin,extra-margin = 0pt, columns-width = 0pt]
        p_{10}p_{21} & p_{11}p_{20} \\ 
        \drawCustomCircletwo{0}{0}{0.25cm}{pink}{pink}{0.8}{0pt}{$p_{20}p_{11}$} & p_{21}p_{10} \\
    \end{bNiceArray}
    \begin{bNiceArray}{>{\strut}cc}[margin,extra-margin = 0pt, columns-width = -1pt]
    \drawCustomCircletwo{0}{0}{0.25cm}{magenta}{magenta}{0.8}{0pt}{} &  
    \drawCustomCircletwo{0}{0}{0.25cm}{magenta}{magenta}{0.8}{0pt}{} \\
    \drawCustomCircletwo{0}{0}{0.25cm}{magenta}{magenta}{0.8}{0pt}{} &  
    \drawCustomCircletwo{0}{0}{0.25cm}{magenta}{magenta}{0.8}{0pt}{} \\
    \end{bNiceArray}
    \\
    V_{4} { \Pi_{2}(\rho_A \otimes \rho_B)^{\modesuper{2}}\Pi_{0}} V_{2} 
    &= 
    \begin{bNiceArray}{>{\strut}c}[margin,extra-margin = 3pt, columns-width = 0pt]
        \drawCustomCircletwo{0}{0}{0.25cm}{violet}{violet}{0.8}{0pt}{}
    \end{bNiceArray} 
    \begin{bNiceArray}{>{\strut}ccc}[margin,extra-margin = 3pt, columns-width = 0pt]
        \drawCustomCircletwo{0}{0}{0.25cm}{pink}{pink}{0.8}{0pt}{$p_{20}p_{22}$}, & p_{21}p_{21}, & p_{22}p_{20} 
    \end{bNiceArray} 
    \begin{bNiceArray}{>{\strut}ccc}[margin,extra-margin = 3pt, columns-width = 8pt]
        \drawCustomCircletwo{0}{0}{0.25cm}{RedOrange}{RedOrange}{0.8}{0pt}{} &  
        \drawCustomCircletwo{0}{0}{0.25cm}{RedOrange}{RedOrange}{0.8}{0pt}{} &  
        \drawCustomCircletwo{0}{0}{0.25cm}{RedOrange}{RedOrange}{0.8}{0pt}{}  \\
        \drawCustomCircletwo{0}{0}{0.25cm}{RedOrange}{RedOrange}{0.8}{0pt}{} &  
        \drawCustomCircletwo{0}{0}{0.25cm}{RedOrange}{RedOrange}{0.8}{0pt}{} &  
        \drawCustomCircletwo{0}{0}{0.25cm}{RedOrange}{RedOrange}{0.8}{0pt}{}  \\
        \drawCustomCircletwo{0}{0}{0.25cm}{RedOrange}{RedOrange}{0.8}{0pt}{} &  
        \drawCustomCircletwo{0}{0}{0.25cm}{RedOrange}{RedOrange}{0.8}{0pt}{} &  
        \drawCustomCircletwo{0}{0}{0.25cm}{RedOrange}{RedOrange}{0.8}{0pt}{}  \\
    \end{bNiceArray}
\end{align}
In the above, we have highlighted what coefficients within each LR-D subspace belong to basis states in {$\mathcal{V}_{\rm in}^{(2)}$} using a pink dot. The aim when performing mode concentration is still to maximise the absolute value of the sum of the coefficients {associated with these basis operators.} What this framework provides is a finer level of detail over how these coefficients can evolve. 

To find the bound given in Result \ref{result:bound_2}, one applies the same methodology as for Result~\ref{result:bound_1}, but this time to each LR-D subspace individually. It can be seen, as above, that an {upper bound} on {$M^{\modesuper{2}}(\sigma_A)$} is given by the sum of the absolute values of the matrix coefficients associated {with basis operators} in {$\mathcal{V}_{\rm in}^{(2)}$.} Using Lemma~\ref{lemma_sum_coefficents}, this can be upper-bounded by sum over the $k$th Ky-Fan norm of all the $(c+j,c)$-LR-D operators, where $k$ is the number of basis states in {$\mathcal{V}_{\rm in}^{(2)}$} upon which a given $(c+j,c)$-LR-D operator has support.    

\end{document}